\PassOptionsToPackage{unicode}{hyperref}
\PassOptionsToPackage{hyphens}{url}
\PassOptionsToPackage{dvipsnames,svgnames,x11names}{xcolor}
\PassOptionsToPackage{format=plain,
  labelfont={bf,sf,small,singlespacing},
  textfont={sf,small,singlespacing},
  justification=justified,
  margin=0.25in}{caption}%
\documentclass[
  11pt,
]{article}
\usepackage{xcolor}
\usepackage[lmargin=1in,rmargin=1in,tmargin=1in,bmargin=1in]{geometry}
\usepackage{amsmath,amssymb}
\usepackage{iftex}
\ifPDFTeX
  \usepackage[T1]{fontenc}
  \usepackage[utf8]{inputenc}
  \usepackage{textcomp} 
    \usepackage[p,osf,swashQ,scale=1.02]{cochineal}
  \usepackage{amsthm}
  \usepackage[cochineal,vvarbb]{newtxmath}
      \usepackage{biolinum}
    \usepackage[scale=0.95,varl]{inconsolata}
\else 
    \usepackage{unicode-math} 
    \usepackage{mathpazo}
  \usepackage{newpxtext}
    \usepackage{biolinum}
    \defaultfontfeatures{Scale=MatchLowercase}
  \defaultfontfeatures[\rmfamily]{Ligatures=TeX,Scale=1}
\fi

\usepackage{sectsty}
\ifPDFTeX\else
\fi
\IfFileExists{upquote.sty}{\usepackage{upquote}}{}
\IfFileExists{microtype.sty}{
  \usepackage[]{microtype}
  \UseMicrotypeSet[protrusion]{basicmath} 
}{}

\usepackage{longtable,booktabs,array}
\usepackage{calc} 
\usepackage{etoolbox}
\makeatletter
\patchcmd\longtable{\par}{\if@noskipsec\mbox{}\fi\par}{}{}
\makeatother
\IfFileExists{footnotehyper.sty}{\usepackage{footnotehyper}}{\usepackage{footnote}}
\makesavenoteenv{longtable}
\usepackage{graphicx}
\makeatletter
\newsavebox\pandoc@box
\newcommand*\pandocbounded[1]{
  \sbox\pandoc@box{#1}%
  \Gscale@div\@tempa{\textheight}{\dimexpr\ht\pandoc@box+\dp\pandoc@box\relax}%
  \Gscale@div\@tempb{\linewidth}{\wd\pandoc@box}%
  \ifdim\@tempb\p@<\@tempa\p@\let\@tempa\@tempb\fi
  \ifdim\@tempa\p@<\p@\scalebox{\@tempa}{\usebox\pandoc@box}%
  \else\usebox{\pandoc@box}%
  \fi%
}
\def\fps@figure{htbp}
\makeatother

\providecommand{\tightlist}{%
  \setlength{\itemsep}{0pt}\setlength{\parskip}{0pt}}

\usepackage[]{natbib}
\usepackage{mathtools}
\usepackage{amsmath}
\usepackage{amsthm}
\usepackage{amssymb}
\usepackage[italicdiff]{physics}

\usepackage[p,osf,swashQ]{cochineal}
\usepackage[cochineal,vvarbb]{newtxmath}
\usepackage{float}
\usepackage{placeins}
\usepackage{setspace}
\newcommand{\papersp}{1.0} 

\usepackage{tikz}
\usetikzlibrary{arrows.meta}
\usepackage[dvipsnames]{xcolor}

\definecolor{sooblue}{HTML}{10577F}
\definecolor{proxred}{HTML}{9A5510}
\definecolor{ivpink}{HTML}{C040BA}
\definecolor{gray}{HTML}{777777}

\newcommand{\soo}{\mathrm{soo}}
\newcommand{\prox}{\mathrm{prox}}
\newcommand{\iv}{\mathrm{iv}}

\newcommand{\bias}{\mathrm{Bias}}

\newcommand{\TRV}{\mathrm{RV}}

\newcommand{\hlsoo}[1]{\textcolor{sooblue}{#1}}
\newcommand{\hlprox}[1]{\textcolor{proxred}{#1}}
\newcommand{\hliv}[1]{\textcolor{ivpink}{#1}}

\newcommand{\biasSOO}{\mathrm{Bias}(\tau_\soo; \hlsoo{R_{Y \sim U \mid Z, W_Y, W_Z}}, \hlsoo{R_{Z \sim U \mid W_Y, W_Z})}}

\newcommand{\sooRsq}{\{\hlsoo{R_{Y \sim U \mid Z, W_Y, W_Z}}, \hlsoo{R_{Z \sim U \mid W_Y, W_Z}}\}}

\newcommand{\R}{\ensuremath{\mathbb{R}}}

\newcommand{\eps}{\varepsilon} 

\DeclareMathOperator{\E}{\mathbb{E}}

\newcommand{\indep}{\mathbin{\perp\!\!\!\!\!\:\perp}} 

\DeclareMathOperator{\Var}{\mathbb{V}}

\DeclareMathOperator{\cov}{cov}
\DeclareMathOperator{\cor}{corr}
\let\var\undefined
\DeclareMathOperator{\var}{var}
\DeclareMathOperator{\sd}{sd}

\makeatletter
\@ifpackageloaded{caption}{}{\usepackage{caption}}
\AtBeginDocument{%
\ifdefined\contentsname
  \renewcommand*\contentsname{Table of contents}
\else
  \newcommand\contentsname{Table of contents}
\fi
\ifdefined\listfigurename
  \renewcommand*\listfigurename{List of Figures}
\else
  \newcommand\listfigurename{List of Figures}
\fi
\ifdefined\listtablename
  \renewcommand*\listtablename{List of Tables}
\else
  \newcommand\listtablename{List of Tables}
\fi
\ifdefined\figurename
  \renewcommand*\figurename{Figure}
\else
  \newcommand\figurename{Figure}
\fi
\ifdefined\tablename
  \renewcommand*\tablename{Table}
\else
  \newcommand\tablename{Table}
\fi
}
\@ifpackageloaded{float}{}{\usepackage{float}}
\floatstyle{ruled}
\@ifundefined{c@chapter}{\newfloat{codelisting}{h}{lop}}{\newfloat{codelisting}{h}{lop}[chapter]}
\floatname{codelisting}{Listing}

\usepackage{amsthm}
\theoremstyle{plain}
\newtheorem{theorem}{Theorem}[section]
\theoremstyle{plain}
\newtheorem{corollary}{Corollary}[section]
\theoremstyle{plain}
\newtheorem{lemma}{Lemma}[section]
\theoremstyle{plain}
\newtheorem{proposition}{Proposition}[section]
\theoremstyle{remark}
\AtBeginDocument{}

\makeatother
\makeatletter
\@ifpackageloaded{caption}{}{\usepackage{caption}}
\@ifpackageloaded{subcaption}{}{\usepackage{subcaption}}
\makeatother
\usepackage{bookmark}
\IfFileExists{xurl.sty}{\usepackage{xurl}}{} 
\hypersetup{
  pdftitle={Relative Bias Under Imperfect Identification in Observational Causal Inference},
  pdfauthor={Melody Huang; Cory McCartan},
  pdfkeywords={sensitivity analysis, proximal inference, instrumental
variables},
  colorlinks=true,
  linkcolor={black},
  filecolor={Maroon},
  citecolor={VioletRed4},
  urlcolor={DodgerBlue4},
  pdfcreator={LaTeX via pandoc}}

\makeatletter
\newcommand{\assumplabel}[2]{%
   \protected@write \@auxout {}{\string\newlabel{#1}{{#2}{\thepage}{#2}{#1}{}}}%
   \hypertarget{#1}{#2}%
}
\makeatother
\newenvironment{assump}[2][]{%
\par\medskip\noindent%
\def\tempa{}\def\tempb{#1}%
\lowercase{\def\lblkey{asm-#2}}%
\textbf{Assumption \assumplabel{\lblkey}{#2}}%
\ifx\tempa\tempb\else~(#1)\fi%
\textbf{.}\itshape%
}{\medskip\normalfont}
\makeatletter
\newcommand\@shorttitle{}
\newcommand\shorttitle[1]{\renewcommand\@shorttitle{#1}}
\usepackage{fancyhdr}
\fancyheadoffset{0pt}
\makeatother
\renewenvironment{abstract}{
  \centerline
  {\large\sffamily\bfseries Abstract}\vspace{-0.25em}
  \begin{quote}\small
}{
  \end{quote}
}

\title{\sffamily\bfseries\huge\parfillskip=0pt
\rightskip=0pt plus .5\textwidth
\leftskip=0pt plus .5\textwidth
\emergencystretch=.3\textwidth Relative Bias Under Imperfect
Identification in Observational Causal Inference}
\shorttitle{Relative Bias Under Imperfect Identification in Observational Causal Inference}
\author{\textbf{Melody Huang}\footnote{
To whom correspondence should be addressed.
Email: \texttt{\href{mailto:melody.huang@yale.edu}{melody.huang@yale.edu}}.
The authors thank Chad Hazlett, Kosuke Imai, Sam Pimentel, Eric Tchetgen
Tchetgen, Carlos Cinelli, Melissa Sands, participants at PolMeth XLII,
the Oxford Causal Inference Reading Group, CU Boulder Van Ek Seminar,
Stanford Quantitative Methods Workshop, Harvard Applied Statistics
Workshop, University of Rochester Riker Seminar, and Berkeley Causal
Inference Working Group for helpful comments and suggestions. The
authors used large language models to auto-formalize results in this
paper to Lean. The statement of each result was reviewed by the authors.
The Lean compiler verifies each formalized proof.}
\\Department of Political Science%
\vspace{2pt}
\\Department of Statistics \& Data Science%
\\Yale University%
\vspace{2pt}
 \and \textbf{Cory McCartan}
\\Department of Statistics%
\\Pennsylvania State University%
\vspace{2pt}
 }
\date{October 13, 2026}

\begin{document}
\allsectionsfont{\sffamily}

\maketitle

\begin{abstract}
To conduct causal inference in observational settings, researchers must
rely on certain identifying assumptions. In practice, these assumptions
are unlikely to hold exactly. This paper considers the bias of
selection-on-observables, instrumental variables, and proximal inference
estimates under violations of their identifying assumptions. We develop
bias expressions for IV and proximal inference that show how violations
of their respective assumptions are amplified by any unmeasured
confounding in the outcome variable. We propose a set of sensitivity
tools in the form of a numeric summary measure that quantifies the
relative sensitivity of different identification strategies, and a
visual summary tool in the form of an augmented bias contour plot that
visualizes the relationship between these strategies. We argue that the
act of choosing an identification strategy implicitly expresses a belief
about the degree of violations that must be present in alternative
identification strategies. Even when researchers intend to conduct an IV
or proximal analysis, a sensitivity analysis comparing different
identification strategies can help to better understand the implications
of each set of assumptions. Throughout, we compare the different
approaches on a re-analysis of the impact of state surveillance on the
incidence of protest in Communist Poland.
\end{abstract}

\textbf{\textit{Keywords}}\quad sensitivity
analysis~\textbullet~proximal inference~\textbullet~instrumental
variables

\clearpage

\section{Introduction}\label{sec-intro}

Causal inference in observational studies requires identification
assumptions which account for potential confounding that occurs from
non-random assignment. The most familiar identification assumption is
that of no unobserved confounding, also known as ignorable treatment
assignment, or \emph{selection on observables} (SOO): given observed
covariates, potential outcomes are independent of treatment. In many
observational settings, however, SOO is considered implausible. The
presence of unobserved confounders can bias treatment effects which are
estimated under an SOO assumption.

Alternative identification strategies have been proposed as an
alternative to the selection-on-observables assumptions. Researchers may
try to find an \emph{instrumental variable} (IV) which is associated
with treatment but satisfies an exclusion restriction with the outcome.
Under certain additional assumptions, an IV allows for estimating the
effect of interest. More recently, \emph{proximal} inference has been
developed as an alternative identification approach
\citep{miao2015identification, miao2018identifying, cui2024semiparametric}.
Instead of assuming researchers measure all relevant confounders,
proximal inference assumes researchers have access to two informative
proxies: a treatment proxy and an outcome proxy. If these proxies
satisfy certain conditional independence assumptions, not unlike the
exclusion restriction in IV, then a treatment effect can be estimated,
even though the proxies themselves are confounded.

The existence of three distinct identification approaches leads to
natural questions about how the corresponding estimators compare in
practice. Much work has been conducted to compare the relative
efficiency of the different approaches \citep[see,
e.g.,][]{andrews2019weak}. In particular, weak instruments, and by
extension, weak proxies, can inflate the variance of IV and proximal
estimators, such that SOO estimators will have lower mean squared error
even in settings when the IV and proximal assumptions hold, and the SOO
assumptions do not.

Fewer studies have examined how the bias of different identification
strategies compares. The nature of identification assumptions is that
they are untestable, and in real-world observational data, it is
unlikely that any set of identifying conditions, SOO, IV, or proximal,
will hold exactly. While recent literature has introduced new
sensitivity analyses to evaluate the sensitivity of estimators to
potential violations in the underlying identifying assumptions
\citetext{\citealp[e.g.,][]{rosenbaum1987sensitivity}; \citealp{frank2000impact}; \citealp{imbens2003sensitivity}; \citealp{tan2006distributional}; \citealp{carnegie2016assessing}; \citealp{ding2019decomposing}; \citealp{zhao2019sensitivity}; \citealp{cinelli2020making}; \citealp{hong2021did}; \citealp{dorn2023sharp}; \citealp[for
SOO]{huang2025variance}; \citealp{kang2021ivmodel}; \citealp{freidling2022optimization}; \citealp[
for IV]{cinelli2025iv}; \citealp[and][ for proximal]{cobzaru2024bias}},
these approaches focus on a single identification strategy, making it
challenging to compare the relative sensitivities across different
approaches.

Our paper introduces a framework that allows researchers to compare the
relative bias of different identification strategies in the realistic
setting where none of the identifying assumptions hold exactly. We are,
to our knowledge, the first to develop a sensitivity framework that
allows for comparisons across identification approaches. Our
contributions are two-fold.

First, in Section~\ref{sec-bias}, we derive bias expressions for SOO,
IV, and proximal estimators in terms of the violations of their
underlying assumptions. While the bias expression for SOO is now
well-known \citep{cinelli2020making}, to our knowledge the
parametrizations for IV and proximal bias are new. We show that the bias
that arises from violations in one set of identification assumptions can
impact the bias in alternative identification approaches. Specifically,
the bias that arises in IV will additionally depend on violations in the
SOO assumptions, while the bias for a proximal estimator will depend on
\emph{both} violations in the IV assumptions and the SOO assumptions.

Second, in Sections \ref{sec-sens-comp} and \ref{sec-unify}, to help
researchers compare the relative sensitivities across the different
identification strategies, we propose a set of sensitivity tools in the
form of numerical and visual summary measures. To compare the
sensitivities of the various identification approaches, we introduce the
\emph{generalized robustness value}, which measures the minimum amount
of total violation of a set of identification assumptions that result in
a certain amount of bias. The generalized robustness value allows
researchers to numerically compare the minimum violation needed to
overturn a research result for each identification strategy. We also
introduce a unified bias contour plot as a visual summary tool. In
particular, we show that the bias of both IV and proximal inference can
be written as a function of just the selection-on-observables parameters
and the data. This unification allows us to visualize the bias of all
three identification strategies on a single contour plot.

Our sensitivity framework can aid practitioners in several ways. Often,
applied researchers identify a causal effect of interest, consider
plausible identification strategies, and then collect data to enable
inference under their chosen strategy. Often, as in our running example,
researchers may conduct a primary analysis under one identification
strategy and then conduct a secondary analyses as a robustness check. In
contrast to existing sensitivity analyses, which work with a single
identification strategy, our proposed measures are constructed to
directly account for the nested dependency in violations across the
different identification assumptions. This allows researchers to
directly compare different identification strategies and choose one most
appropriate for their data.

The proposed framework is still useful, however, for researchers who
have already settled on an identification strategy for their primary
analysis. As both the bias expressions in Section~\ref{sec-bias} and the
unification in Section~\ref{sec-unify} show, the SOO, IV, and proximal
assumptions are deeply interrelated. Even when the SOO assumptions are
not plausible, the value of the SOO \emph{estimate} places certain
restrictions on plausible values for the different sensitivity
parameters. For example, the observed data may imply that \emph{either}
the IV assumptions hold \emph{or} a certain minimum amount of treatment
confounding must exist. If an analyst is confident about the direction
of confounding bias, that may also immediately rule out the plausibility
of certain other identifying assumptions. Thus our framework provides
for better interrogation and integration of identifying assumptions,
allowing researchers to consider their assumptions from a different
angle and move beyond fully \emph{a priori} reasoning about their
plausibility.

\subsection{Running example: Surveillance and Protest in Communist
Poland}\label{sec-intro-ex}

To illustrate our discussion and demonstrate the proposed sensitivity
analyses in a concrete setting, we re-analyze data from
\citet{hk2022poland}, who studied the effects of state surveillance on
resistance to the Polish regime, as measured by group and individual
protests. The authors digitized records from the Służba Bezpieczeństwa
(SB; Department of Security) from 1949--1989, covering the entire
existence of the Polish People's Republic. These records include the
number of secret agents assigned to each municipality each year. The
authors focus on the region of Upper Silesia, where data was
particularly comprehensive. They also collected data on protests
organized by the Solidarność (Solidarity) movement, which was the center
of resistance to the governing regime in the 1980s, as well as indirect
data on individual refusals to work on Saturdays. The research question
is whether additional surveillance of a municipality caused either group
or individual protest to increase or decrease.\footnote{The original
  paper by \citet{hk2022poland} introduces two different outcomes of
  interest. For completeness, we also run our proposed framework on
  their second group protest outcome, which is provided in
  Appendix~\ref{sec-app-sens}. We observe similar patterns for both
  outcomes.}

Hager and Krakowski's primary analysis is based on an instrumental
variable: the number of local Catholic priests who were secretly
cooperating with the SB. The police intentionally recruited priests, who
in turn recruited other local informants. However, priests were randomly
assigned to parishes by the Catholic Church, making the number of
\emph{corrupt} priests plausibly exogenous. The authors include a
detailed study and discussion of the validity of the instrument, which
we also discuss below. As a robustness check, they also conducted a
cross-sectional analysis under a SOO assumption. Both analyses found
that state surveillance increased group protests but decreased
individual protests. Here, we study the sensitivity of both the IV and
SOO findings, and also conduct a sensitivity analysis for proximal
inference using the same data.

\section{Setup and Background}\label{setup-and-background}

This section provides an overview of three common identification
strategies for the ATE in cross-sectional observational settings:
selection-on-observables, instrumental variables, and proximal
inference. We express these strategies in terms of a common statistical
model, which we introduce next, and we discuss the plausibility of the
different identifying assumptions in the context of our running example.

We study the following partially linear structural model for an outcome
\(Y\), treatment \(Z\), vector of covariates \(X\), unobserved
confounder \(U\), and two additional variables \(W_Y\) and
\(W_Z\).\footnote{Because we are studying the behavior of three
  different identification strategies, we require a model for all four
  variables: \(Y\), \(Z\), \(W_Y\), and \(W_Z\). In settings where
  researchers are only interested in selection-on-observables, they only
  need to assume a model for \(Y\). Similarly, for IV, a model for
  \(W_Z\) and \(Y\) are sufficient.}
\begin{equation}\protect\phantomsection\label{eq-struct}{
\begin{aligned}
Y &= \beta_u U + \beta_{w_y} W_Y + \beta_{w_z} W_Z + \tau Z + f_y(X) + \eps_y \\
Z &= \gamma_u U + \gamma_{w_y} W_Y + \gamma_{w_z} W_Z + f_z(X) + \eps_z \\
W_Z &=\varphi_u U + \varphi_{w_y} W_Y + f_{w_z}(X) + \eps_{w_z} \\
W_Y &= \alpha_u U + f_{w_y}(X) + \eps_{w_y} \\
\end{aligned}
}\end{equation} Above, \(\eps_y, \eps_z, \eps_{w_z}\), and
\(\eps_{w_y}\) are noise terms independent of all other variables, and
\(f_y\), \(f_z\), \(f_{w_y}\), and \(f_{w_z}\) are measurable functions
of the covariates \(X\). The estimand of interest is the treatment
effect, represented by the coefficient in front of \(Z\) in the equation
for \(Y\), which we denote \(\tau\).\footnote{The stable unit treatment
  value assumption (SUTVA) is also implied by this setup. In settings
  where the treatment \(Z\) is binary, we can map the structural models
  to the potential outcomes framework, where the potential outcomes for
  \(Y\) under control and treatment, which we write as \(Y(0)\) and
  \(Y(1)\), can be obtained by substituting \(Z=0\) and \(Z=1\) into the
  equation for \(Y\). As a result, the average treatment effect (ATE) is
  \(\E[Y(1) - Y(0)] = \tau\).} This can be thought of as a special case
of the fully non-parametric setting studied in
\citet{miao2018identifying}. For expositional clarity, the rest of the
paper works with the reduced model that (nonparametrically) partials out
\(X\); this also has the effect of centering the remaining variables at
zero. In a slight abuse of notation, we will use the same notation for
the variables and coefficients in the reduced model.
Figure~\ref{fig-dag} illustrates the setup and the assumptions discussed
below. We take a population view of the linear models. Thus, the `hat'
above different terms (i.e., \(\hat Z\) and \(\hat W_Y\)) is used to
denote that these are predicted values (at the population-level), not
that they correspond to sample estimates. \newline 

\emph{Remark.} Throughout the paper, we assume that \(W_Z\) and \(W_Y\)
are pre-treatment variables. In other words, we do not account for
settings when researchers are using post-treatment negative controls for
potential candidates of the treatment proxy \(W_Y\), or settings where
\(W_Z\) (i.e., the instrument or outcome proxy) is potentially
post-treatment. The framework can be adapted to accommodate these
settings, though this would require either accounting for potential
post-treatment bias from erroneously adjusting for these variables in a
selection-on-observables setting, or accounting for different sets of
covariates across the different identification strategies.

\begin{figure}[t]
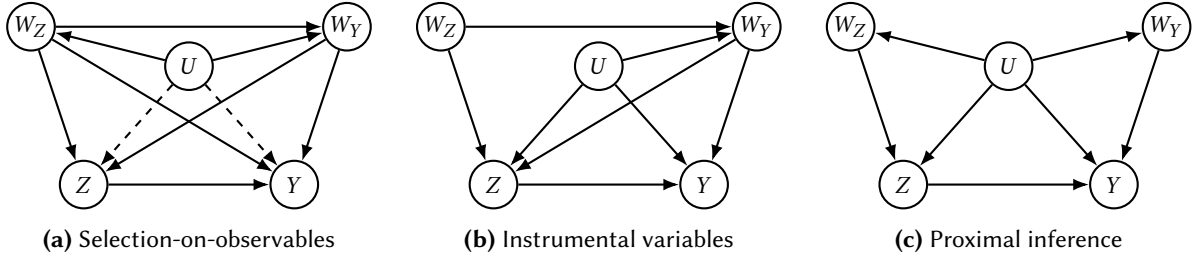


\begin{minipage}[t]{0.33\linewidth}

\centering{

\tikzset{> = {Latex[length=6pt]},
         every path/.append style = {arrows = ->, thick},
         every node/.append style = {inner sep=0, font = \footnotesize}}
\tikz[scale=0.7]{
    \node[draw, circle, minimum size=18pt] (u) at (0, 0.25) {$U$};
    \node[draw, circle, minimum size=18pt] (z) at (-2, -2)  {$Z$};
    \node[draw, circle, minimum size=18pt] (y) at (2, -2) {$Y$};
    \node[draw, circle, minimum size=18pt] (wz) at (-3, 1) {$W_Z$};
    \node[draw, circle, minimum size=18pt] (wy) at (3, 1) {$W_Y$};
    \path[dashed] (u) edge (z);
    \path[dashed] (u) edge (y);
    \path (u) edge (wz);
    \path (u) edge (wy);
    \path (z) edge (y);
    \path (wz) edge (z);
    \path (wy) edge (wz);
    \path (wz) edge (y);
    \path (wy) edge (z);
    \path (wy) edge (y);
}

}

\subcaption{\label{fig-dag-soo}Selection-on-observables}

\end{minipage}%
\begin{minipage}[t]{0.33\linewidth}

\centering{

\tikzset{> = {Latex[length=6pt]},
         every path/.append style = {arrows = ->, thick},
         every node/.append style = {inner sep=0, font = \footnotesize}}
\tikz[scale=0.7]{
    \node[draw, circle, minimum size=18pt] (u) at (0, 0.25) {$U$};
    \node[draw, circle, minimum size=18pt] (z) at (-2, -2)  {$Z$};
    \node[draw, circle, minimum size=18pt] (y) at (2, -2) {$Y$};
    \node[draw, circle, minimum size=18pt] (wz) at (-3, 1) {$W_Z$};
    \node[draw, circle, minimum size=18pt] (wy) at (3, 1) {$W_Y$};
    \path (u) edge (z);
    \path (u) edge (y);
    \path (u) edge (wy);
    \path (z) edge (y);
    \path (wz) edge (z);
    \path (wy) edge (wz);
    \path (wy) edge (z);
    \path (wy) edge (y);
}

}

\subcaption{\label{fig-dag-iv}Instrumental variables}

\end{minipage}%
\begin{minipage}[t]{0.33\linewidth}

\centering{

\tikzset{> = {Latex[length=6pt]},
         every path/.append style = {arrows = ->, thick},
         every node/.append style = {inner sep=0, font = \footnotesize}}
\tikz[scale=0.7]{
    \node[draw, circle, minimum size=18pt] (u) at (0, 0.25) {$U$};
    \node[draw, circle, minimum size=18pt] (z) at (-2, -2)  {$Z$};
    \node[draw, circle, minimum size=18pt] (y) at (2, -2) {$Y$};
    \node[draw, circle, minimum size=18pt] (wz) at (-3, 1) {$W_Z$};
    \node[draw, circle, minimum size=18pt] (wy) at (3, 1) {$W_Y$};
    \path (u) edge (z);
    \path (u) edge (y);
    \path (u) edge (wz);
    \path (u) edge (wy);
    \path (z) edge (y);
    \path (wz) edge (z);
    \path (wy) edge (y);
}

}

\subcaption{\label{fig-dag-prox}Proximal inference}

\end{minipage}%

\caption{\label{fig-dag}The causal structure studied in this paper,
after partialing out any covariates \(X\). Selection-on-observables
assumes that at least one of the dashed arrows is absent.}

\end{figure}%

\subsection{Selection on observables}\label{selection-on-observables}

The ATE \(\tau\) can be estimated from Eq.~\ref{eq-struct} only if \(U\)
is observed, which it is not. The simplest way to estimate \(\tau\) in
the presence of an unobserved \(U\) is to assume that the observed
covariates \(X\) contain all relevant confounders, an assumption known
as selection-on-observables (SOO). In our model, SOO means that at least
one of the following two assumptions holds.

\begin{assump}[Unconfounded treatment]{SOO-Z}
\(Z\indep U\mid X, W_Z, W_Y\).\end{assump}

\begin{assump}[Unconfounded outcome]{SOO-Y}
\(Y\indep U\mid X, Z, W_Z, W_Y\).\end{assump}

Assumption~\ref{asm-soo-z} corresponds to \(\gamma_u=0\) in
Eq.~\ref{eq-struct}, and Assumption~\ref{asm-soo-y} corresponds to
\(\beta_u=0\). Either of these assumptions is sufficient for \(\tau\) to
be identified, and estimated via a regression on the observed variables,
with \(\tau_\soo\) the coefficient on \(Z\) in this regression. The
following standard result, whose proof is omitted, records this fact.

\begin{proposition}[SOO
validity]\protect\hypertarget{prp-soo-valid}{}\label{prp-soo-valid}

If either Assumption~\ref{asm-soo-z} or Assumption~\ref{asm-soo-y}
holds, \(\tau_\soo=\tau\).

\end{proposition}

In practice, the selection-on-observables assumptions cannot be tested
and are often untenable. Researchers are often constrained by the
pre-treatment covariates that are available in observational contexts.
In the presence of latent, mismeasured, or unknown confounders, relying
on a SOO identification strategy will generally result in biased
estimates.

\subsection{Instrumental Variables}\label{instrumental-variables}

Instrumental variables have long been used in contexts where researchers
believe that confounding remains after controlling for observed
covariates. Informally, this identification approach relies on finding a
variable \(W_Z\), the \emph{instrument}, that can explain variation in
the treatment \(Z\), but can only affect the outcome \(Y\) through its
relationship with \(Z\). The variation in \(W_Z\) can then be used to
identify the treatment effect \(\tau\), even in the presence of an
unobserved confounder \(U\). While IV no longer requires Assumptions
\ref{asm-soo-z} or \ref{asm-soo-y}, identification under our model
requires that \(W_Z\) satisfy the following two assumptions.

\begin{assump}[Exogenous instrument]{IV-EXOG}
\(W_Z\indep U\mid X, W_Y\).\end{assump}

\begin{assump}[Exclusion Restriction]{IV-EXCL}
\(Y\indep W_Z\mid X, U, Z, W_Y\).\end{assump}

An instrument \(W_Z\) is considered a \emph{valid instrument} if both
exclusion restriction, as well as exogeneity, holds. In the context of
Eq.~\ref{eq-struct}, Assumption~\ref{asm-iv-exog} corresponds to
\(\varphi_u=0\) and Assumption~\ref{asm-iv-excl} corresponds to
\(\beta_{w_z}=0\). The instrument must also be \emph{relevant}, meaning
that \(\gamma_{w_z}\neq 0\), so it has some effect on \(Z\). Given a
valid instrument, the most common way to estimate the treatment effect
is via two-stage least squares, which will recover \(\tau\). In general
settings beyond the model in Eq.~\ref{eq-struct}, additional assumptions
about treatment effect homogeneity or monotonicity are required to
identify a causal effect.

\begin{proposition}[IV
validity]\protect\hypertarget{prp-iv-valid}{}\label{prp-iv-valid}

Let \(\tau_\iv\) be estimated via two-stage least squares:

\begin{enumerate}
\def\labelenumi{\arabic{enumi}.}
\tightlist
\item
  Regress \(Z\) on \(W_Z\) and \(W_Y\) and generate predicted values
  \(\hat Z\).
\item
  Regress \(Y\) on \(\hat Z\) and \(W_Y\). The least-squares coefficient
  on \(\hat Z\) is denoted \(\tau_\iv\).
\end{enumerate}

\noindent Then, under Assumptions \ref{asm-iv-exog} and
\ref{asm-iv-excl}, \(\tau_\iv=\tau\).\footnote{The derivation of this
  result relies on a homogeneous treatment effect \(\tau\), which we
  have assumed from the partial linearity of the structural model.
  Alternative identification approaches, which we do not study, replace
  homogeneity with other structural assumptions such as monotonicity.}

\end{proposition}

As Section~\ref{sec-intro-ex} describes, in our running example from
\citet{hk2022poland}, the authors used the number of local Catholic
priests who were secretly cooperating with the SB as an instrumental
variable. They show the number of corrupted priests is strongly
correlated with the number of SB officers, and justify both IV
identification assumptions using both substantive knowledge and
auxiliary data. However, like all identification approaches, there still
remain existing mechanisms that could violate either assumption. For
example, if the Catholic Church reserved assignments to important
cities, such as those with a history of protest or surveillance, for
priests that they judged to be more resistant to recruitment, then the
exogeneity assumption would no longer hold. Similarly, the authors note
that the possibility of indoctrination of parishioners by corrupted
priests would result in a violation of exclusion restriction. We provide
more discussion in Section~\ref{sec-app-covar}. All in all, like many
other real-world settings, the complex social dynamics in
\citet{hk2022poland} make it unlikely that both IV assumptions hold
exactly, and warrant investigation of the sensitivity of the estimates
to violations of either assumption.

\subsection{Proximal Inference}\label{proximal-inference}

While IV allows researchers to relax the selection-on-observables
assumptions in exchange for two instrumental variables assumptions,
finding a valid instrument can be challenging. In particular, an
instrument \(W_Z\) must explain variation in \(Z\), but be exogenous.
Recently, \citet{miao2015identification} proposed an alternative
identification strategy, \emph{proximal inference}, which relaxes the IV
assumption of exogeneity.

Proximal inference relies on finding two `proxy variables': (1) an
outcome proxy \(W_Y\) that is unrelated to the treatment but is related
to the outcome, and (2) a treatment proxy \(W_Z\) that is unrelated to
the outcome, but is related to the treatment. The treatment proxy is
similar to an instrumental variable, but does not need to be exogenous.
As such, proximal inference can be viewed as an alternative to
instrumental variables, in settings when researchers are not confident
they have an exogenous instrument.

More formally, the proximal inference approach of
\citet{miao2015identification} replaces SOO or IV assumptions with two
assumptions about the proxy variables.

\begin{assump}[Outcome proxy unrelated to treatment]{PROX-Y}
\(W_Y\indep Z, W_Z\mid X, U\).\end{assump}

\begin{assump}[Treatment proxy unrelated to outcome]{PROX-Z}
\(W_Z\indep Y\mid X, U, Z, W_Y\).\end{assump}

When both Assumptions \ref{asm-prox-y} and \ref{asm-prox-z} hold, we say
that \(W_Y\) and \(W_Z\) are valid proxy variables. In the context of
Eq.~\ref{eq-struct}, Assumption~\ref{asm-prox-y} corresponds to
\(\gamma_{w_y}=0\) and \(\varphi_{w_y}=0\), and
Assumption~\ref{asm-prox-z} means \(\beta_{w_z}=0\). Notice that
Assumption~\ref{asm-prox-z} is exactly the same as
Assumption~\ref{asm-iv-excl}.

Like instrumental variables, a two-stage least squares procedure can be
used to estimate the treatment effect. However, instead of regressing
the treatment indicator \(Z\) with the instrument \(W_Z\), the first
stage relies on regressing the outcome proxy \(W_Y\) with \(Z\) and
\(W_Z\). Because \(W_Y\) has no direct causal connection with \(Z\) and
\(W_Z\) under Assumption~\ref{asm-prox-y}, any correlation between the
two sets of variables can only be due to the unobserved confounder
\(U\), and thus the first stage amounts to estimating a proxy for \(U\).
That estimated \(\hat W_Y\) is then included in the second-stage, along
with \(Z\). When both \(W_Y\) and \(W_Z\) are valid proxies, the
estimated treatment effect under proximal inference (denoted
\(\tau_{\prox}\)) will recover the average treatment effect \(\tau\).

\citet{hk2022poland} did not originally run a causal analysis under
proximal assumptions. Yet, as described above, both SOO and IV
identification strategies face plausible challenges to their validity,
which proximal identification might address. The treatment proxy \(W_Z\)
is the count of corrupted priests, corresponding to the instrument under
IV. For the outcome proxy \(W_Y\), we argue that the indicator for
former Russian occupation is a reasonable candidate. Hager and Krakowski
explain that Russia implemented ``russification'' policies in these
areas ``which were intended to destroy Polish communities and
identities,'' in turn making protest against the Polish regime less
likely. A mechanism for affecting the amount of surveillance is less
obvious; the authors suggest in passing that Russia's experience in the
occupied territories may have made administering spy networks easier,
but we find this mechanism less plausible. Indeed, the occupation
variable has a partial correlation with the group protest outcome about
three times larger than the correlation with the treatment, after
controlling for other covariates, and a ten-times larger correlation
with the individual protest outcome than treatment.

Proximal identification is no panacea, however. Valid inference still
hinges on the exclusion restriction holding, and the outcome proxy must
satisfy two additional exclusion restrictions. Specifically, Russian
occupation cannot directly affect the number of corrupted priests or the
amount of surveillance. While plausible for some of the reasons Hager
and Krakowski outline in their discussion of the exogeneity of the
instrument, it is not unreasonable to assume that the same weakening of
Polish communities in occupied areas might affect which priests are
assigned and their success in recruiting informants. Thus, as with the
other two identification strategies, it remains important to examine the
sensitivity of estimates to violations of the identifying assumptions.

\subsection{Summary and illustration on running
example}\label{summary-and-illustration-on-running-example}

Table \ref{tbl-id} summarizes the three identification approaches, their
assumptions, and the regressions used in estimation. It also displays
the estimated effect of surveillance on individual protests in our
running example, under each estimation approach. The scale of the
effects is difficult to interpret due to several layers of rescaling and
pre-processing in the underlying data. Here, we have rescaled both the
outcome and treatment variables to have unit variance, to ease
comparison and discussion of the estimates. Thus a value of \(0.3\)
below means a one-s.d. increase in surveillance is associated with an
\(0.3\)-s.d. increase in protest.\footnote{The same estimate corresponds
  to a value of \(0.068\) in the original paper.}

We find consistent results to \citet{hk2022poland}, where increased
surveillance corresponds to a decrease in individual protests. Notice
that the estimates are all statistically significant at traditional
levels. One immediately apparent aspect of the estimates in Table
\ref{tbl-id} is that the three approaches disagree on the magnitude of
the effect of individual protest.\footnote{The authors additionally
  analyzed an outcome of `group protests', which is a count of protests
  organized by the Solidarność (Solidarity) movement. We provide the
  additional analysis in the Appendix.}

\setstretch{1.0}

\begin{table}

\caption{\label{tbl-id}Summary of identification approaches, with estimates from the running example. Sample size is 216 for the group outcome and 206 for the individual outcome due to missing outcome data.}

\centering{

\centering
\begin{tabular}{>{\raggedright\arraybackslash}p{2.7cm}p{4cm}p{4cm}p{4cm}} \toprule
\textbf{} & \textbf{SOO} (Figure \ref{fig-dag-soo}) & \textbf{IV} (Figure \ref{fig-dag-iv}) & \textbf{Proximal} (Figure \ref{fig-dag-prox}) \\
\midrule
Identification assumptions
  & Either \ref{asm-soo-z} \textit{or} \ref{asm-soo-y}
  & \ref{asm-iv-exog} and \ref{asm-iv-excl}
  & \ref{asm-prox-z} and \ref{asm-prox-y} \\

Implied coefficients
  & $\beta_u \gamma_u = 0$
  & $\beta_{w_z}=\varphi_u=0$
  & $\gamma_{w_y}=\varphi_{w_y}=\beta_{w_z}=0$ \\

Estimation
  & $Y \sim Z + W_Z + W_Y$
  & \begin{tabular}[t]{@{}l@{}}
     TSLS: \\
     1. $Z \sim W_Z + W_Y$ \\
     2. $Y \sim \hat{Z} + W_Y$
    \end{tabular}
  & \begin{tabular}[t]{@{}l@{}}
     TSLS: \\
     1. $W_Y \sim Z + W_Z$ \\
     2. $Y \sim Z + \hat{W_Y}$
    \end{tabular} \\


Estimate (s.e.) for individual protest
  & $-0.296\ \ (0.051)$
  & $-0.805\ \ (0.108)$
  & $-0.627\ \ (0.068)$ \\ \bottomrule
\end{tabular}

}

\end{table}%

\setstretch{\papersp}

\section{Bias decompositions}\label{sec-bias}

This section presents decompositions of the bias that arises from
violations in any of the underlying identification assumptions. We find
that identification strategies are more sensitive to violations of their
assumptions when there is more unobserved confounding. In other words,
violations in Assumption~\ref{asm-soo-y} will \emph{amplify} any
violations in the alternative identification assumptions.

The results here are derived for the population estimates. Because our
model (Eq.~\ref{eq-struct}) is linear, however, they hold equally when
all coefficients are interpreted as their finite-sample versions.
Throughout the paper, we follow the notation in
\citet{cinelli2020making} and use superscripted \(\perp\) to indicate
partialling out, such that for any variables \(A\), \(B\), and \(C\),
\(A^{\perp B, C} := A - \E[A\mid B, C].\) We define
\(\sd(A) := \sqrt{\Var(A)}\) and
\(\cor(A, B) := \mathrm{Cov}(A, B)/(\sd(A)\sd(B))\). We make use of
\emph{partial correlations},
\(R_{A\sim B\mid C} := \cor(A^{\perp C}, B^{\perp C})\), and the
resulting \emph{partial \(R^2\) values} (i.e.,
\(R^2_{A\sim B\mid C} = \cor^2(A^{\perp C}, B^{\perp C})\)). Finally, we
will use \hlsoo{blue} to highlight parameters related to
selection-on-observables assumptions, \hliv{pink} for parameters related
to instrumental variables assumptions, and \hlprox{brown} for parameters
related to proximal inference assumptions.

\subsection{Review: Bias of selection-on-observables
estimator}\label{sec-bias-soo}

For selection-on-observables to be a valid identification approach,
there cannot be an omitted confounder \(U\) that affects both the
outcome \(Y\) and the treatment assignment process \(Z\). In the
presence of an omitted \(U\), the bias for \(\tau_\soo\) will be a
function of \(U\)'s relationship with \(Y\) and \(Z\), respectively.
This can be expressed in terms of the coefficients in
Eq.~\ref{eq-struct} as
\begin{equation}\protect\phantomsection\label{eq-bias-soo}{
\bias(\tau_\soo) = \hlsoo{\beta_u} \cdot \hlsoo{\gamma_u} \cdot \frac{\var(U^{\bot W_Z, W_Y})}{\var(Z^{\bot W_Z, W_Y})}.
}\end{equation} Lemma~\ref{lem-soo} contains the full derivation.

Consider a potential confounder of \emph{state capacity} in the running
example. Then, \(\hlsoo{\beta_u}\) represents the relationship between
state capacity and the outcome (either group or individual protest) that
is not already captured by the observed proxies of state capacity (i.e.,
number of schools). Similarly, \(\hlsoo{\gamma_u}\) represents the
residual imbalance of state capacity across the treatment and control
groups, after controlling for the observed covariates. If researchers
believe that proxies of state capacity sufficiently capture the
variation in state capacity, then we expect \(\hlsoo{\beta_u}\) and
\(\hlsoo{\gamma_u}\) to be relatively small in magnitude, and little
bias to result for not controlling for the latent variable of state
capacity. In contrast, researchers could be worried that the measurement
error in using the number of schools as a proxy of state capacity could
be related to protests, in which case, we might expect there to be more
bias from omitting a variable like state capacity.

Importantly, Eq.~\ref{eq-bias-soo} formalizes the well-studied result
that if omitted confounder \(U\) is related to \emph{both} the outcome
process \emph{and} the treatment assignment process, there will be bias
\citep{imbens2003sensitivity}. Even if \(U\) is imbalanced between the
treatment and control groups, if it does not explain variation in the
outcome, then omitting it will not lead to bias, since \(\beta_u = 0\).
Similarly, if \(U\) explains variation in the outcome but does not
explain variation in the treatment assignment process, then
\(\gamma_u = 0\). While it may be advantageous to control for \(U\) for
precision, doing so is not strictly necessary to eliminate bias.

\subsection{Bias of instrumental variables}\label{sec-bias-iv}

We now turn to the bias of instrumental variables. Under violations of
exogeneity and exclusion restriction, we can express the bias in
\(\tau_\iv\) as a function of three key terms: (a) a scaling factor that
depends on the strength of the instrument, (b) bias from violations in
the exclusion restriction, as captured by \(\hliv{\beta_{w_z}}\), and
(c) the bias from violations in instrument exogeneity, captured by the
term \(\hliv{\varphi_u}\).

\begin{theorem}[Bias of IV
estimate]\protect\hypertarget{thm-iv}{}\label{thm-iv}

The bias of \(\tau_\iv\) is given by
\begin{equation}\protect\phantomsection\label{eq-iv-bias-r2}{
\begin{aligned}
\bias(\tau_{\iv}) &= \underbrace{\frac{1}{R_{Z\sim W_Z\mid W_Y}} \cdot
    \frac{\sd(W_Z^{\bot W_Y})}{\sd(Z^{\bot W_Y})}}_{\text{Observable Scaling Factor}} \left\{\hliv{\beta_{w_z}}  + \hlsoo{\beta_u}  \cdot \hliv{\varphi_u} \cdot \frac{\var(U^{\bot W_Y})}{\var(W_Z^{\bot W_Y})}\right\}
\end{aligned}
}\end{equation} where \(\beta_{w_z}, \beta_u, \varphi_u\) correspond to
the parameters in Eq.~\ref{eq-struct}.

\end{theorem}

All proofs are deferred to the appendix. We have also verified all
results in this paper with Lean \citep{lean4}; formalized results are
available in the replication archive.

Theorem~\ref{thm-iv} formalizes the different drivers of bias in an IV
estimate. The first term is a scaling factor, which is comprised of
fully observable quantities. While the scaling factor does not depend on
violations in the underlying IV assumptions, it will amplify any
potential violations. Notably, the scaling factor depends on the
strength of the underlying instrumental variable as measured by
\(R_{Z\sim W_Z\mid W_Y}\). In settings when there is a weak instrument,
\(R_{Z\sim W_Z\mid W_Y}\) will be low, and even in the presence of very
weak violations in the underlying assumptions, the overall bias incurred
by an IV estimator will be relatively large. This is on top of any
inferential challenges, such as high variance and non-normality, that
arise from weak instruments even when the IV assumptions hold exactly
\citep{andrews2019weak}.

The second term is \(\hliv{\beta_{w_z}}\), which measures how related
the instrument \(W_Z\) is to the outcome \(Y\), after conditioning on
\(\{W_Y, Z, U\}\). This is a direct representation of violations in the
exclusion restriction.

Finally, the last term represents the bias incurred from violations of
exogeneity and is represented by the product between
\(\hlsoo{\beta_{u}}\), \(\hliv{\varphi_u}\), and a partially
unobservable scaling factor \(\var(U^{\bot W_Y})/\var(W_Z^{\bot W_Y})\).
The coefficient \(\hliv{\varphi_u}\) represents the residual variation
in \(W_Z\) that can be explained by the omitted confounder \(U\), after
controlling for observed covariates, and directly maps to the violation
of Assumption~\ref{asm-iv-exog}. Interestingly, \(\hliv{\varphi_u}\) is
scaled by \(\hlsoo{\beta_u}\), which corresponds to violations in the
underlying SOO assumptions. We can informally think of the SOO
parameters as controlling the overall `difficulty' of the identification
problem. When \(U\) is a relatively weak confounder such that
\(\hlsoo{\beta_u}\) is close to zero, then there will be less exposure
to bias for \emph{both} \(\tau_\soo\) and \(\tau_\iv\). In this setting,
\(\tau_\iv\) will be more robust to violations of exogeneity.
Conversely, more outcome confounding directly increases the effect of
violations of exogeneity.

The bias decomposition highlights that there are dependencies between
the bias of the \(\tau_\iv\) and the selection-on-observables
assumptions. In other words, even though \(\tau_\iv\) relies on
different identification assumptions as \(\tau_\soo\), the bias in
\(\tau_\iv\) can be amplified by violations in the
selection-on-observables assumptions. Appendix~\ref{sec-app-exog}
further discusses the special case when researchers are willing to
assume exogeneity of the instrument holds exactly, and wish to examine
sensitivity to the exclusion restriction. \newline 

\noindent \textbf{Remark.} Recently, \citet{cinelli2025iv} developed a
sensitivity analysis for instrumental variables in a linear model based
on an inversion of an Anderson--Rubin test
\citep{anderson1949estimation}. Their approach allows for a sensitivity
analysis that incorporates sampling uncertainty and applies to upper and
lower confidence bounds, even in the presence of weak instruments. This
allows researchers to understand the minimum amount of confounding that
would change an IV estimate to no longer be statistically significant.
The resulting sensitivity analysis is parameterized in terms of two
partial correlations, one of which involves an (unobservable) potential
outcome. Here, in contrast, we focus solely on the effect of assumption
violations on the bias of the IV estimate. This leads us to directly
consider the 2SLS estimate and parameterize the bias in terms of
coefficients involving the observed variables and \(U\) only, without
potential outcomes. Theorem~\ref{thm-iv} is a helpful complement to the
sensitivity analysis of \citet{cinelli2025iv} in understanding the
behavior of IV estimates.

\subsection{Bias of proximal estimate}\label{sec-bias-prox}

Finally, we turn to the bias of the proximal estimator \(\tau_\prox\).
Like the results in Section~\ref{sec-bias-iv}, the bias of
\(\tau_\prox\) will implicitly depend on not only the degree to which
the proximal assumptions are violated, but also the
selection-on-observables and IV assumptions.

\begin{theorem}[Bias of the proximal
estimator]\protect\hypertarget{thm-prox-bias}{}\label{thm-prox-bias}

The bias of \(\tau_\prox\) is: \[
\bias(\tau_\prox)= S_{W_Z, Z}\cdot \hliv{\beta_{w_z}} \frac{\sd(W_Z^{\bot \hat W_Y})}{\sd(Z^{\bot \hat W_Y})} +  \hlsoo{\beta_u} (\hlprox{\varphi_{w_y}}\hlsoo{\gamma_u} - \hlprox{\gamma_{w_y}} \hliv{\varphi_u})\frac{ \left( 1- R^2_{W_Y \sim U} \right) \var(W_Y) \var(U)}{\cov(W_Y, W_Z) \var(Z) - \cov(W_Y, Z) \cov(W_Z, Z)} 
\] where
\(S_{W_Z, Z}\coloneq \cor(W_Z^{\perp \hat W_Y}, Z^{\perp \hat W_Y})\) is
either \(1\) or \(-1\).

\end{theorem}

The theorem formalizes that the bias of a proximal estimator is going to
be made up of two parts: the first corresponds to the bias from an
invalid treatment proxy, and the second corresponds to bias from an
invalid outcome proxy. Both sources of bias are scaled by values that
correspond to the strength of the respective proxies.

To start, the first driver of bias results from an invalid treatment
proxy, as measured by \(\hliv{\beta_{w_z}}\). This parameter is
identical to the parameter from the IV estimator, which corresponded to
the degree to which the exclusion restriction has been violated. This
clarifies that while proximal estimators allow researchers to relax the
exogeneity assumption associated with an instrumental variable,
violations of the exclusion restriction impact both approaches
similarly. This term is scaled by a term that depends on how much
variation in the constructed \(\hat W_Y\) can be attributed to both
\(W_Z\) and \(Z\).

The second driver of bias arises from an invalid outcome proxy. A valid
outcome proxy requires that \(W_Y\) is not related to both \(Z\) and
\(W_Z\), after conditioning on \(U\) (and other pre-treatment
covariates). This is represented by \(\hlprox{\gamma_{w_y}}\) and
\(\hlprox{\varphi_{w_y}}\), which measure how related \(W_Y\) is to
\(Z\) and \(W_Z\), respectively, controlling for both \(U\) and other
pre-treatment covariates. Violations from an invalid outcome proxy are
amplified by \emph{both} violations in selection-on-observables (i.e.,
\(\hlsoo{\gamma_u}\)) and violations in exogeneity (i.e.,
\(\hliv{\varphi_u}\)).

\subsection{Scale-invariant representations of
bias}\label{scale-invariant-representations-of-bias}

We can equivalently re-write the bias expressions derived in
Theorem~\ref{thm-iv} and Theorem~\ref{thm-prox-bias} in terms of partial
correlations instead of regression coefficients and residual variances.
Doing so results in a representation of the bias that is invariant to
the scale of the underlying variables. In Appendix
Section~\ref{sec-app-r2}, we present the bias expressions in terms of
partial correlations. Table \ref{tbl-r2} provides a summary of the
different coefficients and the corresponding partial \(R^2\) values.

\begin{table}

\caption{\label{tbl-r2}Summary of different identification assumptions, the corresponding coefficients in the structural model, and the partial $R^2$ terms that measure the degree of violation of each assumption.}

\centering{

\centering 
\begin{tabular}{>{\raggedright\arraybackslash}p{4cm}p{2cm}p{2cm}p{6.5cm}} \toprule
Assumption & $R^2$ & Coefficient & Interpretation \\ \midrule 
\multicolumn{4}{l}{\textbf{Selection-on-observables}} \\
Selection-on-observables & $\hlsoo{R_{Z \sim U \mid W_Z, W_Y}}$ & $\hlsoo{\gamma_u}$ & $U$'s relationship with the treatment \\
&  $\hlsoo{R_{Y \sim U \mid Z, W_Z, W_Y}}$ & $\hlsoo{\beta_u}$ & $U$'s relationship with the outcome \\ \midrule 
\multicolumn{4}{l}{\textbf{Instrumental variables}} \\
Exclusion restriction & $\hliv{R_{Y \sim W_Z \mid Z, W_Y, U}}$ & $\hliv{\beta_{w_z}}$  & Relationship between instrument and outcome, controlling for $Z$, $W_Y$, and $U$ \\
Exogeneity & $\hliv{R_{W_Z \sim U \mid W_Y}}$ & $\hliv{\varphi_u}$ & Relationship between instrument and $U$, controlling for $W_Y$ \\\midrule 
\multicolumn{4}{l}{\textbf{Proximal inference}} \\ 
Outcome proxy & $\hlprox{R_{Z \sim W_Y \mid W_Z, U}}$ & $\hlprox{\gamma_{w_y}}$ & Relationship between outcome proxy and treatment, controlling for $W_Z$ and $U$ \\
& $\hlprox{R_{W_Z \sim W_Y \mid U}}$ & $\hlprox{\varphi_{w_y}}$ & Relationship between outcome proxy and treatment proxy, controlling for $U$ \\
Treatment proxy & $\hliv{R_{Y \sim W_Z \mid Z, W_Y, U}}$ & $\hliv{\beta_{w_z}}$ & Relationship between treatment proxy and outcome, controlling for $Z$, $W_Y$, and $U$\\ \bottomrule 
\end{tabular} 

}

\end{table}%

While re-writing the bias in terms of partial correlation terms allows
for a representation of the bias that is represented by scale-invariant
parameters, the bias expressions are challenging to reason about in
practice. The bias results in the previous section, Theorem~\ref{thm-iv}
and Theorem~\ref{thm-prox-bias}, highlight how violations of the
different identifying assumptions interact and lead to bias in estimates
of \(\tau\). This is further exacerbated by the \(R^2\) representations.
In particular, these coefficients and variances are not
variation-independent---the range of possible values for a coefficient
depends on the values of other coefficients. Additionally, it is
difficult to make comparisons \emph{across} identification strategies,
because each bias expression is parameterized in terms of a different
set of coefficients and variances. In the following section, we address
these limitations by reparameterizing the problem in terms of four
unitless partial correlations, three of which directly measure the
violation of certain identifying assumptions.

\section{Comparing the sensitivity of identification
approaches}\label{sec-sens-comp}

The bias expressions presented so far depend on numerous coefficients or
partial correlations, and these may not be variation-independent of each
other: setting one coefficient or correlation to a value may restrict
the range of possible values for other coefficients or correlations,
which complicates the interpretation of the bias expressions. In this
section, we introduce a variation-independent parameterization of the
bias for all three of the identification strategies. We show that the
variation independent parameterization can be directly mapped back to
the partial \(R^2\) values that measure the degree of violation of the
different identifying assumptions. Doing so allows us to define a
numerical summary measure called the \emph{robustness value} that
quantifies the minimum violation of the identifying assumptions needed
to result in a substantively meaningful change in the result. The
proposed measure generalizes the robustness value introduced in
\citet{cinelli2020making} to the IV and proximal case, and allows
researchers to compare the sensitivities across the different
identification approaches. To aid in the interpretation of the RV, we
propose a benchmarking approach, which allows researchers to reason
about the plausibility of potential violations in their underlying
identification assumptions.

\subsection{Re-parameterizing the bias using the population covariance
matrix}\label{re-parameterizing-the-bias-using-the-population-covariance-matrix}

Any approach to comparing identification strategies requires
parametrizing the estimates, the bias, and the identifying assumptions
on a common space. To start, because each estimator is simply the
coefficient from a linear regression, we can write each estimator in
terms of the population covariance matrix \(\Sigma \in \R^{7 \times 7}\)
of \((Y, Z, W_Z, W_Y, \hat Z, \hat W_Y, U)\): \[
\tau_\soo =  \Sigma_{zw_zw_y,zw_zw_y}^{-1}\Sigma_{zw_zw_y,y}, \quad
\tau_\iv =   \Sigma_{\hat zw_y,\hat zw_y}^{-1}\Sigma_{\hat zw_y,y}, \qand
\tau_\prox = \Sigma_{z\hat w_y,z\hat w_y}^{-1}\Sigma_{z\hat w_y,y},
\] where subscripts here indicate the submatrix of \(\Sigma\) that
corresponds to the variables in the subscript; i.e.,
\(\Sigma_{zw_zw_y,y}\) is the \(3\)-by-\(1\) submatrix containing the
covariance of \((Z, W_Z, W_Y)\) with \(Y\). Similarly, the true effect
\(\tau\) can be expressed as
\(\tau = \Sigma_{zw_zw_yu,zw_zw_yu}^{-1}\Sigma_{zw_zw_yu,y}\).

While the different estimators do not depend on \(U\), the true effect
\(\tau\) depends on \(U\). Furthermore, because the observed data fixes
most of the matrix \(\Sigma\), only the submatrix
\(\Sigma_{u,yzw_zw_y}\) is unidentified.\footnote{ Because \(\hat Z\)
  and \(\hat W_y\) are deterministic functions of observed variables
  only, the submatrix \(\Sigma_{u,\hat z\hat w_y}\) will depend
  deterministically on the other observed entries as well as
  \(\Sigma_{u,yzw_zw_y}\), leaving only the submatrix
  \(\Sigma_{u,yzw_zw_y}\) as being truly unidentified.} Thus, another
view of sensitivity analysis is to vary the entries in the submatrix
\(\Sigma_{u,yzw_zw_y}\) to evaluate how \(\tau\) changes.

The four entries in \(\Sigma_{u,yzw_zw_y}\) are not unconstrained---they
must be such that \(\Sigma\) is positive semidefinite. Without loss of
generality, we can set \(\Sigma_{uu}=1\), since the scale of the
confounder does not affect any observables. We then parameterize
\(\Sigma_{u,yzw_zw_y}\) by a set of partial correlations of \(U\) and
each observable variable: \[
\rho \coloneq (R_{U\sim W_Y}, \hliv{R_{U\sim W_Z\mid W_Y}},  \hlsoo{R_{U\sim Z\mid W_Y,W_Z}}, \hlsoo{R_{U\sim Y\mid Z,W_Z,W_Y}}).
\] Notice that the latter two partial correlations are exactly the SOO
sensitivity parameters, and the second,
\(\hliv{R_{U\sim W_Z\mid W_Y}}\), measures the violation of
Assumption~\ref{asm-iv-exog}. The first parameter, \(R_{U\sim W_Y}\),
measures the strength of the outcome proxy, but does not directly
measure the violation of any identifying assumption.

Critically, this choice of partial correlations admits a bijective
mapping to the Cholesky decomposition of \(\Sigma\) and thus to
\(\Sigma\) itself \citep{lewandowski2009generating}. As a result, unlike
the entries in \(\Sigma_{u,yzw_zw_y}\), these parameters are variation
independent: any value of \(\rho\in[-1, 1]^4\) corresponds to a valid
covariance matrix \(\Sigma(\rho)\).

The mapping \(\rho\mapsto\Sigma\) fixes the true effect \(\tau\). It
also determines the value of the various partial correlations that
measure how much each identifying assumption is violated. In other
words, we can re-express each of the partial correlation values that
correspond to the different violations in identifying assumptions as
functions of \(\rho\). For example, the violation of the exclusion
restriction corresponds to the value of
\(\hliv{R_{Y\sim W_Z\mid Z,W_Y,U}}\), and can be written as \[
\hliv{R_{Y\sim W_Z\mid Z,W_Y,U}}(\rho) = \frac{\Sigma(\rho)^{-1}_{y,w_z}}{
\sqrt{\Sigma(\rho)^{-1}_{y,y}\Sigma(\rho)^{-1}_{w_z,w_z}}}.
\] Thus, even though the exclusion restriction violation
\(\hliv{R_{Y\sim W_Z\mid Z,W_Y,U}}\) does not appear anywhere in
\(\rho\), unlike the SOO sensitivity parameters, it is still related
nonlinearly to the SOO parameters and the other entries in \(\rho\), due
to the common parametrization.

\subsection{Robustness Value}\label{robustness-value}

The reparametrization of the bias in terms of \(\rho\) enables the
definition of the \emph{robustness value}, one measure of the
sensitivity of an estimate to violations of identifying assumptions. The
robustness value was first introduced by \citet{cinelli2020making} as
the minimum violation of the identifying assumptions to result in a
substantively meaningful change in the underlying research conclusions.

To begin, for a fixed \(\rho\), we define the vector of parameters
associated with each identification strategy as \[
\begin{aligned}
A_\soo(\rho) &= \qty( \hlsoo{R_{Z\sim U\mid W_Z,W_Y}(\rho)},~\hlsoo{R_{Y\sim U\mid Z,W_Z,W_Y}(\rho)} ), \\
A_\iv(\rho) &= \qty( \hliv{R_{W_Z\sim U\mid W_Y}(\rho)},~\hliv{R_{Y\sim W_Z\mid Z,W_Y,U}(\rho)} ), \qand \\
A_\prox(\rho) &= \qty( \hliv{R_{Y\sim W_Z\mid Z,W_Y,U}(\rho)},~\hlprox{R_{Z\sim W_Y\mid  W_Z,U}(\rho)},~ \hlprox{R_{W_Z\sim W_Y\mid U}(\rho)} ).
\end{aligned}
\] Because the SOO sensitivity parameters are part of \(\rho\) itself,
\(A_\soo\) is simply comprised of the last two entries in \(\rho\);
however, \(A_\iv\) and \(A_\prox\) are nonlinear functions of \(\rho\).
While the bias of both the IV and proximal estimators additionally
depend on other partial correlations, \(A_\iv(\rho)\) and
\(A_\prox(\rho)\) represent the partial correlations that directly map
to the underlying identification assumptions.

We can then formally define the \emph{robustness value} (RV) as the
minimum bound on the violations of the identifying assumption that is
sufficient to guarantee absolute bias no more than some threshold amount
of bias \(b\): \begin{equation}\protect\phantomsection\label{eq-rv}{
\begin{aligned}
\TRV_\soo(b) &= \inf \big\lvert\big\lvert
    A_\soo(\{\rho\in P\,:\, \lvert \bias(\tau_\soo(\rho)) \rvert =b \})
    \big\rvert\big\rvert^2_\infty, \\
\TRV_\iv(b) &= \inf \big\lvert\big\lvert
    A_\iv(\{\rho\in P\,:\, \lvert \bias(\tau_\iv(\rho)) \rvert=b \})
    \big\rvert\big\rvert^2_\infty, \qand \\
\TRV_\prox(b) &= \inf \big\lvert\big\lvert
    A_\prox(\{\rho\in P\,:\, \lvert \bias(\tau_\prox(\rho)) \rvert =b \})
    \big\rvert\big\rvert^2_\infty.
\end{aligned}
}\end{equation} Here, \(P\subseteq [-1,1]^4\) is a possibly restricted
set of allowable values of \(\rho\), which we discuss below.

The robustness value depends on an analyst-chosen bias threshold \(b\),
representing the smallest change in a causal estimate that is considered
meaningful. A common choice is \(b\) equal to the estimate itself, or
the estimate plus or minus two standard errors, and thus large enough to
reverse the sign of the causal effect. Alternatively, \(b\) could be
chosen to be equal to some multiple of the estimate's standard error.
For a Normal sampling distribution, a bias of one standard error reduces
the coverage of nominally 95\% confidence intervals to around 80\%; a
two-s.e. bias reduces coverage to around 50\%. This is similar to the
notion of a `killer confounder', first presented in
\citet{huang2024sensitivity} and \citet{hartman2024sensitivity}.

Eq.~\ref{eq-rv} captures the notion that violations smaller than the RV
guarantee bias smaller than the critical threshold. The robustness value
measures a kind of safety guarantee: \emph{if} all of the assumptions
are violated less than the RV, as measured by the partial \(R^2\)
values, \emph{then} the bias will be smaller than \(b\). A
\emph{smaller} robustness value means that a smaller violation of the
identifying assumptions could be sufficient to result in a threshold
bias of \(b\). In contrast, a larger robustness value means that a
larger violation of the identifying assumptions is needed to obtain the
same bias.

Because of the common parametrization, it is possible to compare the
robustness values of different identification approaches. In an extreme
case, if the RV for one identification strategy is very small, while it
is relatively large for the others, then this means that a smaller
violation in the assumptions of that strategy could result in a
substantial amount of bias; this can warrant caution in evaluating the
credibility of the estimate from that strategy. However, each RV
measures violations of different assumptions, and so when RV values are
closer, direct comparisons may be more difficult. We develop a
benchmarking procedure for the RV based on observed covariates in
Section~\ref{sec-benchmarking} to aid in these comparisons.

The proposed robustness value generalizes the robustness value defined
in \citet{cinelli2020making}, which was introduced for the
selection-on-observables setting.\footnote{\citet{cinelli2020making}
  define the RV through a specific formula, but motivate it in the way
  defined here, and for the SOO setting that they studied, their
  explicit definition and the implicit one in Eq.~\ref{eq-rv} coincide.}
However, some properties of the RV for SOO do not carry over to other
strategies. In the SOO setting, the solution to the minimization problem
in Eq.~\ref{eq-rv} requires that the two SOO sensitivity parameters (the
last two components of \(\rho\)) be equal. That may not be the case with
the RV when applied to IV or proximal strategies. Another difference is
that the SOO bias is monotonically increasing in each parameter. For IV
and proximal, however, the addition and subtraction in their bias
expressions means it is possible for assumption violations to cancel
out. Thus, if an assumption is violated more than \(\TRV(b)\), there is
no guarantee that the bias will exceed \(b\) for IV and proximal,
depending on the signs of the partial correlations.

One difficulty with the RV for IV and proximal is the behavior of the
bias when treatment confounding
\(\lvert R_{Z\sim U\mid W_Z,W_Y} \rvert\) (i.e., \(\rho_3\)) is very
high (near 1). In particular, when
\(\lvert R_{Z\sim U\mid W_Z,W_Y} \rvert\) is near 1, the bias of both IV
and proximal can be extremely sensitive to small changes in the other
parameters, even if those parameters are relatively small in magnitude.
In other words, any minuscule change in other parameters in \(\rho\)
results in extreme changes in the value of \(\tau\). This can be seen
visually in a traditional sensitivity contour plot where the contours of
bias are extremely dense around the point \((1, 0)\). The high
derivative of the bias in this region may lead to smaller robustness
values. We show this pattern in Figure~\ref{fig-rv-contour} for our
running example by carrying out the optimization in Eq.~\ref{eq-rv}
while fixing \(R_{Z\sim U\mid W_Z,W_Y}\) to different values along its
possible range. The RV for SOO is minimized at an intermediate value of
\(R_{Z\sim U\mid W_Z,W_Y}\) (i.e., exactly where
\(R^2_{Z\sim U\mid W_Z,W_Y}=R^2_{Y\sim U\mid Z, W_Z,W_Y}\)). However,
the RV for IV and proximal tends to 0 as
\(R_{Z\sim U\mid W_Z,W_Y}\to 1\). This also highlights that in settings
when there is a large degree of confounding in the treatment assignment
process, the bias of IV and proximal estimates will be extremely
sensitive to any violations in their underlying identification
assumptions.

This motivates restricted choices of \(P\), the allowable values of
\(\rho\) in Eq.~\ref{eq-rv}. Restricting \(R_{Z \sim U \mid W_Z, W_Y}\)
to values away from this region avoids this pathology. As a default, we
recommend restricting \(R_{Z \sim U \mid W_Z, W_Y}^2<0.95\), as in most
applied settings it is unlikely that a single confounder could explain
more than 95\% of treatment variation. This restriction still allows for
a large degree of confounding in the treatment assignment process, while
avoiding the region where the bias is extremely sensitive to small
changes in \(\rho\). Fortunately, we have found in examining plots like
Figure~\ref{fig-rv-contour} for other outcomes, and for other DGPs
generated by sampling \(\Sigma\) uniformly from the set of all
correlation matrices, that patterns like the one in
Figure~\ref{fig-rv-contour} are common: while there is variation in the
RV for IV and proximal across values of \(R_{Z \sim U \mid W_Z, W_Y}\),
the relative positions of the RV values for the different identification
strategies are relatively stable. Researchers are unlikely to
incorrectly conclude, for instance, that their IV estimate is more or
less robust than their SOO estimate due to a restriction on
\(R_{Z \sim U \mid W_Z, W_Y}\).

\begin{figure}

\centering{

\pandocbounded{\includegraphics[keepaspectratio]{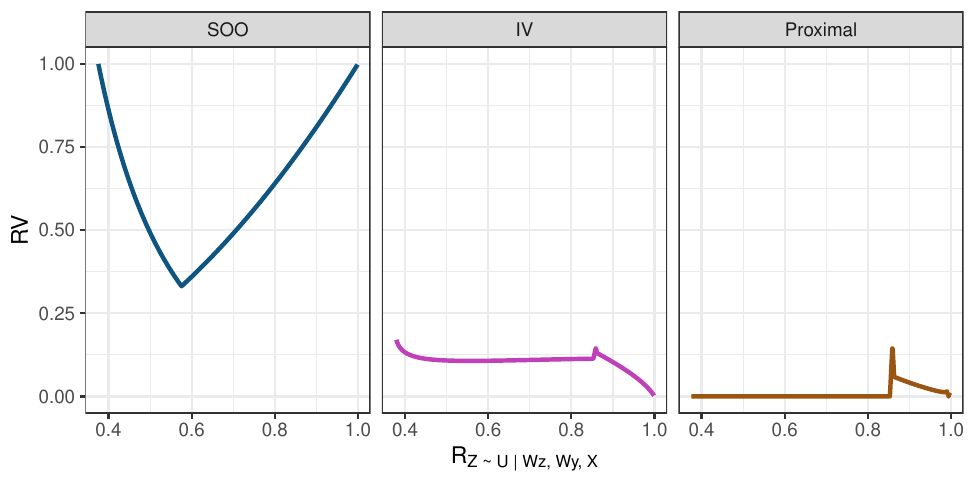}}

}

\caption{\label{fig-rv-contour}Robustness value by identification
strategy for running example as \(\rho_3 = R_{Z\sim U\mid W_Z, W_Y}\) is
constrained to different values along its possible range. The overall
robustness value is the minimum along each curve. The parameter
\(\rho_4=R_{Y\sim U\mid Z, W_Z, W_Y}\) decreases monotonically from 1 to
0 as \(\rho_3\) increases along the \emph{x}-axis.}

\end{figure}%

To estimate the RV in practice, researchers can use numerical
optimization with sample estimates of the covariance matrix. Our
parametrization \(\Sigma(\rho)\) is differentiable, and so researchers
can use reverse-mode automatic differentiation to calculate exact
gradients for the minimization problem in Eq.~\ref{eq-rv}. The bias
equality constraint can be handled by penalization. Researchers can
quantify estimation uncertainty by bootstrapping, as we demonstrate
below.

The RV provides a helpful, single number summary of how much the
underlying assumptions can be violated before risking a substantively
meaningful change in our research conclusion, one that accounts for the
nested dependency structure across the different identification
assumptions. However, it cannot be used to determine which
identification strategy has less (or more) bias necessarily, because the
actual value of \(\rho\) is unknown. However, a small RV \emph{does}
signify that even a small violation in any underlying assumption could
result in a substantial amount of bias, and warrants potential caution
in evaluating the credibility of an estimated causal effect.

\subsection{Benchmarking}\label{sec-benchmarking}

While the RV allows researchers to quantify and compare sensitivities,
it measures different assumptions for each identification strategy,
across which comparison may not always be straightforward. One approach
to aid in these comparisons is \emph{benchmarking}, which allows
researchers to use the observed covariates as hypothetical unobserved
confounders, to calibrate plausible values for the partial correlations
parametrizing assumption violations. This strategy cannot be directly
applied to the partial correlations in Section~\ref{sec-bias}, because
the unobserved confounder \(U\) is often part of the conditioning
set.\footnote{ When \(X\) and \(U\) are both in the conditioning set,
  treating a single covariate \(X_j\) as the confounder \(U\) does not
  change the value of the partial correlation at all.} The shared
population covariance matrix parameterization here, however, allows us
to benchmark the sensitivity parameters and the robustness value for all
the identification strategies.

For the \(j\)-th covariate (where \(j \in \{1, ..., |X|\}\)), we define
\[
\hat \rho^{(j)} = \qty(
    \hat R_{X^{(j)} \sim W_Y \mid X^{-(j)}},
    \hat R_{X^{(j)} \sim W_Z \mid W_Y, X^{-(j)}},
    \hat R_{X^{(j)} \sim Z \mid W_Y, W_Z, X^{-(j)}},
    \hat R_{X^{(j)} \sim Y \mid Z, W_Z, W_Y, X^{-(j)}}
),
\] where \(X^{-(j)}\) is the set of covariates except the \(j\)-th. In
this way, \(\hat \rho^{(j)}\) represents a set of partial correlations
that would correspond to an omitted confounder \(U\) with equivalent
confounding strength as an observed covariate \(X^{(j)}\). Then, because
each partial correlation term in the \(A_\iv\) and \(A_\prox\) is a
function of the elements of \(\rho\), we can directly estimate the
benchmarked parameters, as well as the corresponding maximum assumption
violation that would occur from omitting a confounder \(U\) with
equivalent confounding strength to \(X^{(j)}\) (denoted by
\(\norm{A(\hat \rho^{(j)})}^2_\infty\)). The benchmarked maximum
assumption violation can be compared to the RV to evaluate how much
stronger (or weaker) a hypothetical \(U\) would have to be to result in
an assumption violation that would saturate the minimum violation
strength represented by the RV.

Benchmarking is most useful in settings when researchers have strong
substantive priors about the types of covariates that are prognostic of
different aspects of the data generating process. In general, we caution
that while benchmarking allows researchers to quantify the relative
strength of a confounder in terms of observed covariates, these types of
statements cannot be used to \emph{rule out} the existence of
confounding.

\subsection{Illustration on running example}\label{sec-rv-demo}

To see the proposed sensitivity summary measures in action, we return to
our running example from \citet{hk2022poland}, and calculate robustness
values for the individual protest outcome under each identification
strategy. For the bias threshold \(b\), we use the respective point
estimates, so that a bias of \(b\) would move the estimated effect to
zero. This choice of \(b\) favors proximal and especially IV, since
their estimates are larger in magnitude.\footnote{ For readers concerned
  about the use of a different \(b\) across strategies, the appendix
  contains an analysis of the group protest outcome, where fortuitously
  all estimates, and thus all \(b\), are nearly equal. Here, the larger
  \(b\) for IV and proximal means the analysis is more conservative in
  favor of those strategies, and even so, their calculated RVs are much
  smaller than for SOO.}

Table~\ref{tbl-rv} shows these bias thresholds and their corresponding
robustness values for each identification strategy. A detailed version
of the table in the appendix also shows the minimizing value of \(\rho\)
from the optimization routine. To account for uncertainty, we use a
fractional-weighted bootstrap and report the median absolute deviation
of the estimates across bootstrap iterations. The fractional-weighted
bootstrap, a kind of Bayesian bootstrap, avoids collinearity issues when
projecting out indicator variables that are only present in a few
observations \citep{xu2020applications}, and the use of MAD is
appropriate given the heavy skew of some of the bootstrap distributions,
which results from the quantities being bounded below by \(0\).

\setstretch{1.0}

\begin{longtable}[]{@{}lrrl@{}}
\caption[]{Robustness values for the reanalysis of \citet{hk2022poland},
individual protest outcome. The median absolute deviations are shown in
parentheses.}\label{tbl-rv}\tabularnewline
\toprule\noalign{}
Strategy & \(b\) & \(\TRV(b)\) & (std. err.) \\
\midrule\noalign{}
\endfirsthead
\toprule\noalign{}
Strategy & \(b\) & \(\TRV(b)\) & (std. err.) \\
\midrule\noalign{}
\endhead
\bottomrule\noalign{}
\endlastfoot
\textbf{SOO} & \(-0.296\) & \(0.331\) &
\(\textcolor{gray}{(0.0902)}\) \\
\textbf{IV} & \(-0.805\) & \(0.0173\) &
\(\textcolor{gray}{(1.43\times 10^{-4})}\) \\
\textbf{Proximal} & \(-0.627\) & \(0.011\) &
\(\textcolor{gray}{(0.00358)}\) \\
\end{longtable}

\setstretch{\papersp}

In the example, we see that selection-on-observables has the largest
robustness value across the identification strategies, with
\(\TRV_\soo=0.331\). In other words, if the unobserved confounder
explains less than \(33\%\) of the variation in both treatment and
outcome, then the bias would not exceed \(b=-0.296\). In comparison, the
robustness values for IV and proximal inference are much smaller. For
IV, \(\TRV_\iv=0.0173\) means that a violation of the exclusion
restriction and exogeneity as small as \(0.0173\) on the \(R^2\) scale
is sufficient to create bias as large as \(b=-0.805\). For proximal,
\(\TRV_\prox=0.011\) also leaves little room for even small violations
of the proximal assumptions. This does not necessarily mean that there
is a greater amount of bias in \(\tau_\iv\) and \(\tau_\prox\). However,
it \emph{does} mean that \emph{any} realistic deviation from the
instrumental variables or proximal assumptions could result in a
substantial amount of bias, and researchers should be cautious when
drawing causal conclusions from \(\tau_\iv\) and \(\tau_\prox\).

To help reason about whether it is plausible that the magnitude of
violations captured by the RVs could occur, we perform the benchmarking
discussed above, and estimate the maximum assumption violation that
would occur for an omitted \(U\) that is as strong as an observed
covariate \(X^{(j)}\). Table~\ref{tbl-bench} shows that across the
different covariates, the maximum assumption violations for IV and
proximal are consistently larger than the benchmarked maximum assumption
violations for selection-on-observables. More specifically, SOO's RV is
around \(7\) times larger than the largest benchmarked maximum
assumption violation, whereas the RVs for IV and proximal are around
\(5\)--\(9\) times \emph{smaller} than the benchmarked maximum
assumption violations.

\setstretch{1.0}

{

\begin{longtable}[]{@{}lrrrr@{}}

\caption{\label{tbl-bench}Benchmarking values for individual protest
outcome. Shown are the change in \(\tau\) from leaving out each
covariate, and the benchmarked maximum assumption violations. Median
absolute deviations for each quantity are shown in parentheses.}

\tabularnewline

\toprule\noalign{}
Covariate & Change in \(\tau\) &
\(\norm{A_\soo(\hat \rho(j))}_\infty^2\) &
\(\norm{A_\iv(\hat \rho(j))}_\infty^2\) &
\(\norm{A_\prox(\hat \rho(j))}_\infty^2\) \\
\midrule\noalign{}
\endhead
\bottomrule\noalign{}
\endlastfoot
Wealth & \(9.0\times 10^{-4}~~\textcolor{gray}{(0.005)}\) &
\(0.0020~~\textcolor{gray}{(0.003)}\) &
\(0.093~~\textcolor{gray}{(0.01)}\) &
\(0.093~~\textcolor{gray}{(0.1)}\) \\
Schools & \(-8.8\times 10^{-4}~~\textcolor{gray}{(0.001)}\) &
\(0.0019~~\textcolor{gray}{(0.003)}\) &
\(0.092~~\textcolor{gray}{(0.01)}\) &
\(0.092~~\textcolor{gray}{(0.1)}\) \\
Industrialization & \(-0.00123~~\textcolor{gray}{(0.002)}\) &
\(0.0017~~\textcolor{gray}{(0.003)}\) &
\(0.092~~\textcolor{gray}{(0.01)}\) &
\(0.092~~\textcolor{gray}{(0.1)}\) \\
Log population & \(-0.02884~~\textcolor{gray}{(0.046)}\) &
\(0.0450~~\textcolor{gray}{(0.018)}\) &
\(0.093~~\textcolor{gray}{(0.01)}\) &
\(0.093~~\textcolor{gray}{(0.1)}\) \\
Diversity & \(-0.00707~~\textcolor{gray}{(0.008)}\) &
\(0.0158~~\textcolor{gray}{(0.010)}\) &
\(0.094~~\textcolor{gray}{(0.02)}\) &
\(0.094~~\textcolor{gray}{(0.1)}\) \\
Katowickie region & \(-0.00672~~\textcolor{gray}{(0.025)}\) &
\(0.0394~~\textcolor{gray}{(0.017)}\) &
\(0.100~~\textcolor{gray}{(0.02)}\) &
\(0.100~~\textcolor{gray}{(0.1)}\) \\
Opolskie region & \(0.00421~~\textcolor{gray}{(0.004)}\) &
\(0.0102~~\textcolor{gray}{(0.008)}\) &
\(0.090~~\textcolor{gray}{(0.01)}\) &
\(0.090~~\textcolor{gray}{(0.1)}\) \\

\end{longtable}

}

\setstretch{\papersp}

\section{Implications of choosing an identification
strategy}\label{sec-unify}

The common parametrization from Section~\ref{sec-sens-comp} has deeper
implications for comparing estimators. In particular, when researchers
choose an identification strategy, they are implicitly expressing a
belief about the degree of violations that must be present in
alternative identification strategies. In this section, we show that all
three estimators' biases can be re-expressed in terms of the standard
SOO sensitivity parameters
\(\{\hlsoo{R_{Y \sim U \mid Z, W_Z, W_Y}}, \hlsoo{R_{Z \sim U \mid W_Z, W_Y}}\}\).
This means, for example, that if a researcher believes that the IV
identification assumptions hold, then a certain magnitude of unobserved
confounding must also be present. In some cases, this degree of
confounding may be implausible. Researchers might then conclude that the
IV assumptions cannot hold exactly. The bias re-expression also allows
us to visualize the relative bias of each approach on a single bias
contour plot.

\subsection{Re-writing bias with respect to selection-on-observables
parameters}\label{re-writing-bias-with-respect-to-selection-on-observables-parameters}

Section~\ref{sec-sens-comp} parameterized the data-generating process in
terms of four partial correlations, and used this parametrization to
define and evaluate a set of numerical sensitivity summary measures.
However, for the purposes of evaluating bias in \(\tau_\soo\),
\(\tau_\iv\), and \(\tau_\prox\), we can reduce the number of parameters
to just two.

Specifically, we can re-write the bias of \(\tau_\iv\) and
\(\tau_\prox\) with the selection-on-observables parameters
\(\hlsoo{R_{Y\sim U\mid Z,W_Z,W_Y}}\) and
\(\hlsoo{R_{Z\sim U\mid W_Z,W_Y}}\). The key observation is that for a
particular \(\sooRsq\), the bias of the selection-on-observables
estimate will be fixed, which implicitly also fixes the value of
\(\tau\), the true treatment effect. As such, we can re-express the bias
of \(\tau_\iv\) and \(\tau_\prox\) as
\begin{equation}\protect\phantomsection\label{eq-bias-unif}{
\begin{aligned}
\bias(\tau_\iv) &= \tau_\iv - \underbrace{ \qty{ \tau_\soo -
    \bias(\tau_\soo; \hlsoo{R_{Y\sim U\mid Z,W_Z,W_Y}}, \hlsoo{R_{Z\sim U\mid W_Z,W_Y}}) }}_{\coloneq \tau} \\
\bias(\tau_\prox) &= \tau_\prox - \qty{ \tau_\soo -
    \bias(\tau_\soo; \hlsoo{R_{Y\sim U\mid Z,W_Z,W_Y}}, \hlsoo{R_{Z\sim U\mid W_Z,W_Y}}) }
\end{aligned}
}\end{equation} Both of these expressions depend only on the observed
data and the value \(\biasSOO\), which is fixed by the parameters
\(\sooRsq\). Note that the partial correlations corresponding to other
assumptions, like the exclusion restriction
\(R_{Y\sim W_Z\mid Z,W_Y,U}\) \emph{do} depend on other parts of
\(\rho\); it is only the bias of each estimate that can be reduced to
exactly the SOO parameters.

This unification is advantageous for several reasons. First,
Eq.~\ref{eq-bias-unif} gives us a route to translate assumptions made
for one identification strategy into the language of treatment and
outcome confounding. Suppose a researcher believes that their IV
strategy is highly plausible, perhaps in part due to the IV-specific
sensitivity analyses developed in this paper, and furthermore they are
able to estimate the likely direction of IV bias. As a result, they
might believe that most likely \(0<\bias(\tau_\iv)<0.1\).
Eq.~\ref{eq-bias-unif} then connects this belief to an equivalent belief
about likely values of \(\sooRsq\). The researcher's judgement of the
plausibility of these values can further inform their understanding of
the IV bias, or even their choice of identification strategy.

Second, Eq.~\ref{eq-bias-unif} can be directly solved for the set of SOO
sensitivity parameters where each identification strategy has lower bias
than the others. Because Eq.~\ref{eq-bias-unif} is linear in all bias
terms, each estimator will dominate the others on a particular interval
of values of \(\tau\), corresponding in turn to a particular interval of
values of \(\biasSOO\). These intervals are exactly obtained by cutting
the real line at the midpoints between the values of \(\tau_\soo\),
\(\tau_\iv\), and \(\tau_\prox\). For example, if \(\tau_\soo=0.1\),
\(\tau_\iv=0.4\), and \(\tau_\prox=0.2\), then SOO dominates for the
\(\sooRsq\) corresponding to any true \(\tau<0.15\), proximal inference
dominates sensitivity parameters yielding \(0.15<\tau<0.3\), and IV
dominates for parameters giving \(\tau>0.3\). We can map these intervals
to the traditional sensitivity plot, discussed in the next section, as
contour lines that bound regions where each estimator dominates. These
regions can be visualized together, providing insight into the
assumptions and sensitivity of each identification approach.

Readers will note that the manipulations in Eq.~\ref{eq-bias-unif} are
in no way reliant on using the SOO sensitivity parameters specifically:
the same re-expression could be applied in terms of the partial
correlations for IV or proximal bias. However, the bias expressions for
IV and proximal are substantially more complicated, and rely on
unobservable partial correlations with no direct connection to the IV
and proximal assumptions. Moreover, the subtraction that is present in
those bias expressions means that bias is not monotonic in either
sensitivity parameter the way that it is for SOO, which further
complicates interpretation and visualization. Finally, treatment and
outcome confounding are in many ways the core challenges of causal
inference, and researchers are familiar with considering confounding
mechanisms in these terms. Thus we focus specifically on the SOO
parametrization, as we anticipate it will be easiest to reason about for
applied researchers.

\subsection{Unified bias contour
plots}\label{unified-bias-contour-plots}

The unified bias framework also allows us to extend the standard SOO
sensitivity plot and allow researchers to visually compare the relative
biases across the identification strategies. Traditionally, bias contour
plots have been employed to evaluate SOO bias by varying the two
sensitivity parameters \(\sooRsq\). A contour line corresponding to
\(\tau=0\) is often plotted to help identify which values of sensitivity
parameters would change the sign of the estimate. The benchmarked
sensitivity parameters are often also plotted to help researchers reason
about what could be plausible magnitudes of the underlying sensitivity
parameters.

\begin{figure}

\centering{

\includegraphics[width=6.25in,height=\textheight,keepaspectratio]{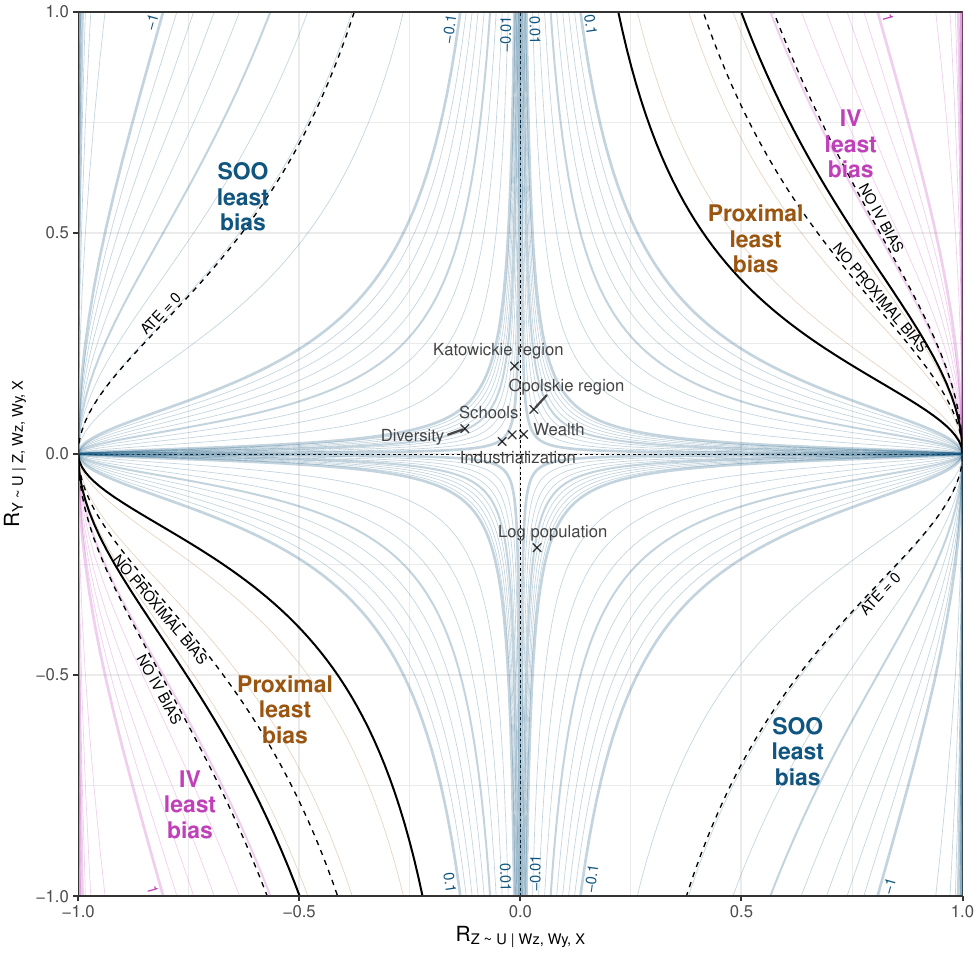}

}

\caption{\label{fig-contour}Sensitivity plot for the running example.
The contour lines are on a logarithmic scale and show the amount of SOO
confounding bias. Contours are colored according to the regions where
each estimator has lowest bias. Gray crosses indicate benchmark
sensitivity parameter values for each of the covariates.}

\end{figure}%

The sensitivity plot we propose builds on this existing visualization in
three ways. First, because the sign of the bias, and thus the dominating
estimator, depends on the sign of the sensitivity parameters, we
generate the sensitivity plot using the signed parameters rather than
their square. Second, as described in the preceding section, we separate
the contour lines according to the regions where each identification
strategy dominates, indicating this partition by lines and colors.
Third, we plot a contour line that corresponds to zero bias in either
the proximal or IV estimate. If the assumptions of either of those
strategies hold exactly, then the SOO sensitivity parameters \emph{must}
lie along the corresponding contour.

While our primary focus throughout the paper has been on the bias of
different identification strategies, researchers may instead be
interested in the relative mean-squared error of the different
approaches. Notably, the same re-expression in Eq.~\ref{eq-bias-unif}
can be applied for mean-squared error, and a similar augmented
sensitivity plot can be produced. In some cases this may show that the
zero-bias contour for proximal or IV is actually contained in the
lowest-MSE region for another estimator. In other words, that \emph{even
if} the IV or proximal assumptions hold exactly, another estimator may
be preferred due to lower mean-squared error. This phenomenon has been
observed for IV estimates when the instrument is weak \citep[see,
e.g.,][]{bound1995problems, hahn2005estimation}. See
Appendix~\ref{sec-app-mse} for more discussion.

\subsection{Illustration on running
example}\label{illustration-on-running-example}

We now return to our running example, where we generate the proposed
augmented sensitivity plot, including the benchmarked values of
\(\sooRsq\) based on the observed covariates (see
Figure~\ref{fig-contour}). If an analyst had originally designed their
observational study around IV, for instance, then their beliefs about
the values of the SOO sensitivity parameters must lie close to the curve
labeled ``No IV Bias'' in Figure~\ref{fig-contour}. In other words,
\emph{belief} about an identification strategy implies a relatively
narrow set of beliefs on the values of the sensitivity parameters.
Practitioners can compare that set of parameter values to the benchmark
values of the covariates, which are plotted as gray crosses in
Figure~\ref{fig-contour}. In this case, the benchmark values lie
squarely in the region labeled ``SOO least bias.'' The largest
benchmarked \(R_{Z\sim U\mid W_Z,W_Y}\) is \(0.0378\) and corresponds to
(log) municipality population. In contrast, the \emph{smallest} possible
value of \(R_{Z\sim U\mid W_Z,W_Y}\) consistent with a valid proximal
strategy is \(0.413\)--around \(11\) times larger. For IV, the factor
grows to \(15\). Both of those factors assume worst-case outcome
confounding, with \(R_{Y\sim U\mid Z,W_Z,W_Y}=1\), which is also several
times larger than the confounding associated with any of the observed
covariates.

Thus if a researcher is to proceed with a proximal or IV strategy, they
must believe that the unobserved confounder is significantly stronger
than any of the observed covariates' benchmark values, and that the
treatment and outcome confounding are related in a way as to fall close
to one of the no-bias contour lines in Figure~\ref{fig-contour}. One
immediate qualitative implication is that treatment and outcome
confounding must have the same sign for a proximal or IV approach to
produce less bias. If they do not, then the sensitivity parameters lie
in quadrants II or IV of Figure~\ref{fig-contour}, where the SOO
estimate always dominates.

Figure~\ref{fig-contour} additionally provides insight on the
generalized robustness values in Section~\ref{sec-sens-comp}. Near the
origin, where the SOO assumptions approximately hold, the contour lines
are sparsely spaced (note that the plotted contours in
Figure~\ref{fig-contour} are on a logarithmic scale). Thus, even
moderate violations of the SOO assumptions lead to small changes in the
value of \(\tau\), and thus small amounts of bias. In contrast, the ``No
Proximal Bias'' and ``No IV Bias'' lines in Figure~\ref{fig-contour} are
much farther from the origin, where the contour lines are denser.
Smaller changes in the sensitivity parameters in these regions lead to
correspondingly larger changes in bias, which in turn lowers the
generalized robustness values.

From Eq.~\ref{eq-bias-unif}, the distance of the ``No Proximal Bias''
and ``No IV Bias'' lines from the origin is directly related to the
difference between the point estimates for these identification
strategies and the SOO estimate. The more the estimates differ, the more
confounding is necessarily implied under each set of identification
assumptions. As both Theorem~\ref{thm-iv} and
Theorem~\ref{thm-prox-bias} show, the bias of the IV and proximal
estimates grows with the amount of outcome confounding. Thus a larger
observed difference between the SOO estimate and IV or proximal
estimates implies that the outcome confounding is larger, which in turn
will generally translate to lower robustness.

\section{Conclusion}\label{sec-conclude}

This paper proposes a unified sensitivity framework for researchers to
consider the relative bias of different identification strategies. While
the unified model imposes some constraints compared to a fully
nonparametric setup, the relative simplicity also clarifies the key
drivers of bias under IV and proximal identification, and highlights
connections between each estimator and their biases. As
Section~\ref{sec-sens-comp} and Section~\ref{sec-unify} show, the
dimensionality of unidentified parts of the causal model is limited, and
this creates (usually nonlinear) relationships between the various
causal assumptions. Thus, by comparing benchmarking values and point
estimates under each strategy, more can be said about each of the
identifying assumptions. The bias expressions in Section~\ref{sec-bias}
and the proposed robustness values also help researchers understand the
sensitivity of each approach to violations of these assumptions.

There are several interesting avenues of future work. One natural
direction would be to explore the extent to which these findings hold in
a fully nonparametric setting. As \citet{chernozhukov2024long} show for
SOO, key intuition and structure from the linear setting can carry over
with little modification to a nonparametric model. The two-stage
estimation of IV and proximal, however, complicates such an analysis and
would warrant careful investigation.

Second, recent work in observational causal inference has introduced an
idea known as design sensitivity for observational studies, which allows
researchers to consider sensitivity to potential violations in the
underlying identification assumptions \emph{a priori}
\citep[e.g.,][]{rosenbaum2004design, rosenbaum2010design, huang2025design}.
Design sensitivity quantifies the tradeoffs in robustness that arise
from different `design' choices (i.e., choosing an estimand, treatment
definitions, or study populations). However, much of this work focuses
explicitly on the SOO setting. Future work could build on the unified
sensitivity framework presented in this paper to consider design
sensitivity across different identification strategies.

\section*{References}\label{references}
\addcontentsline{toc}{section}{References}

\renewcommand{\bibsection}{}
\bibliography{references.bib}

\begin{thebibliography}{}

\bibitem[Anderson and Rubin, 1949]{anderson1949estimation}
Anderson, T.~W. and Rubin, H. (1949).
\newblock Estimation of the parameters of a single equation in a complete system of stochastic equations.
\newblock {\em The Annals of mathematical statistics}, 20(1):46--63.

\bibitem[Andrews et~al., 2019]{andrews2019weak}
Andrews, I., Stock, J.~H., and Sun, L. (2019).
\newblock Weak instruments in instrumental variables regression: Theory and practice.
\newblock {\em Annual Review of Economics}, 11(1):727--753.

\bibitem[Bound et~al., 1995]{bound1995problems}
Bound, J., Jaeger, D.~A., and Baker, R.~M. (1995).
\newblock Problems with instrumental variables estimation when the correlation between the instruments and the endogenous explanatory variable is weak.
\newblock {\em Journal of the American statistical association}, 90(430):443--450.

\bibitem[Carnegie et~al., 2016]{carnegie2016assessing}
Carnegie, N.~B., Harada, M., and Hill, J.~L. (2016).
\newblock Assessing sensitivity to unmeasured confounding using a simulated potential confounder.
\newblock {\em Journal of Research on Educational Effectiveness}, 9(3):395--420.

\bibitem[Chernozhukov et~al., 2024]{chernozhukov2024long}
Chernozhukov, V., Cinelli, C., Newey, W., Sharma, A., and Syrgkanis, V. (2024).
\newblock Long story short: Omitted variable bias in causal machine learning.
\newblock Technical report, National Bureau of Economic Research.

\bibitem[Cinelli and Hazlett, 2020]{cinelli2020making}
Cinelli, C. and Hazlett, C. (2020).
\newblock Making sense of sensitivity: Extending omitted variable bias.
\newblock {\em Journal of the Royal Statistical Society Series B: Statistical Methodology}, 82(1):39--67.

\bibitem[Cinelli and Hazlett, 2025]{cinelli2025iv}
Cinelli, C. and Hazlett, C. (2025).
\newblock An omitted variable bias framework for sensitivity analysis of instrumental variables.
\newblock {\em Biometrika}, page asaf004.

\bibitem[Cobzaru et~al., 2024]{cobzaru2024bias}
Cobzaru, R., Welsch, R., Finkelstein, S., Ng, K., and Shahn, Z. (2024).
\newblock Bias formulas for violations of proximal identification assumptions in a linear structural equation model.
\newblock {\em Journal of Causal Inference}, 12(1):20230039.

\bibitem[Cui et~al., 2024]{cui2024semiparametric}
Cui, Y., Pu, H., Shi, X., Miao, W., and Tchetgen~Tchetgen, E. (2024).
\newblock Semiparametric proximal causal inference.
\newblock {\em Journal of the American Statistical Association}, 119(546):1348--1359.

\bibitem[de~Moura and Ullrich, 2021]{lean4}
de~Moura, L. and Ullrich, S. (2021).
\newblock The lean 4 theorem prover and programming language.
\newblock In {\em Automated Deduction – CADE 28: 28th International Conference on Automated Deduction, Virtual Event, July 12–15, 2021, Proceedings}, page 625–635, Berlin, Heidelberg. Springer-Verlag.

\bibitem[Ding et~al., 2019]{ding2019decomposing}
Ding, P., Feller, A., and Miratrix, L. (2019).
\newblock Decomposing treatment effect variation.
\newblock {\em Journal of the American Statistical Association}, 114(525):304--317.

\bibitem[Dorn and Guo, 2023]{dorn2023sharp}
Dorn, J. and Guo, K. (2023).
\newblock Sharp sensitivity analysis for inverse propensity weighting via quantile balancing.
\newblock {\em Journal of the American Statistical Association}, 118(544):2645--2657.

\bibitem[Frank, 2000]{frank2000impact}
Frank, K.~A. (2000).
\newblock Impact of a confounding variable on a regression coefficient.
\newblock {\em Sociological Methods \& Research}, 29(2):147--194.

\bibitem[Freidling and Zhao, 2022]{freidling2022optimization}
Freidling, T. and Zhao, Q. (2022).
\newblock Optimization-based sensitivity analysis for unmeasured confounding using partial correlations.
\newblock {\em arXiv preprint arXiv:2301.00040}.

\bibitem[Hager and Krakowski, 2022]{hk2022poland}
Hager, A. and Krakowski, K. (2022).
\newblock Does state repression spark protests? evidence from secret police surveillance in communist poland.
\newblock {\em American Political Science Review}, 116(2):564--579.

\bibitem[Hahn and Hausman, 2005]{hahn2005estimation}
Hahn, J. and Hausman, J. (2005).
\newblock Estimation with valid and invalid instruments.
\newblock {\em Annales d'Economie et de Statistique}, pages 25--57.

\bibitem[Hartman and Huang, 2024]{hartman2024sensitivity}
Hartman, E. and Huang, M. (2024).
\newblock Sensitivity analysis for survey weights.
\newblock {\em Political Analysis}, 32(1):1--16.

\bibitem[Hong et~al., 2021]{hong2021did}
Hong, G., Yang, F., and Qin, X. (2021).
\newblock Did you conduct a sensitivity analysis? a new weighting-based approach for evaluations of the average treatment effect for the treated.
\newblock {\em Journal of the Royal Statistical Society: Series A (Statistics in Society)}, 184(1):227--254.

\bibitem[Huang and Pimentel, 2025]{huang2025variance}
Huang, M. and Pimentel, S.~D. (2025).
\newblock Variance-based sensitivity analysis for weighting estimators results in more informative bounds.
\newblock {\em Biometrika}, 112(1):asae040.

\bibitem[Huang et~al., 2025]{huang2025design}
Huang, M., Soriano, D., and Pimentel, S.~D. (2025).
\newblock Design sensitivity and its implications for weighted observational studies.
\newblock {\em Journal of the Royal Statistical Society Series A: Statistics in Society}, page qnaf067.

\bibitem[Huang, 2024]{huang2024sensitivity}
Huang, M.~Y. (2024).
\newblock Sensitivity analysis for the generalization of experimental results.
\newblock {\em Journal of the Royal Statistical Society Series A: Statistics in Society}, 187(4):900--918.

\bibitem[Imbens, 2003]{imbens2003sensitivity}
Imbens, G.~W. (2003).
\newblock Sensitivity to exogeneity assumptions in program evaluation.
\newblock {\em American Economic Review}, 93(2):126--132.

\bibitem[Kang et~al., 2021]{kang2021ivmodel}
Kang, H., Jiang, Y., Zhao, Q., and Small, D.~S. (2021).
\newblock Ivmodel: an r package for inference and sensitivity analysis of instrumental variables models with one endogenous variable.
\newblock {\em Observational Studies}, 7(2):1--24.

\bibitem[Lewandowski et~al., 2009]{lewandowski2009generating}
Lewandowski, D., Kurowicka, D., and Joe, H. (2009).
\newblock Generating random correlation matrices based on vines and extended onion method.
\newblock {\em Journal of multivariate analysis}, 100(9):1989--2001.

\bibitem[Miao et~al., 2018]{miao2018identifying}
Miao, W., Geng, Z., and Tchetgen~Tchetgen, E.~J. (2018).
\newblock Identifying causal effects with proxy variables of an unmeasured confounder.
\newblock {\em Biometrika}, 105(4):987--993.

\bibitem[Miao et~al., 2015]{miao2015identification}
Miao, W., Liu, L., Tchetgen, E.~T., and Geng, Z. (2015).
\newblock Identification, doubly robust estimation, and semiparametric efficiency theory of nonignorable missing data with a shadow variable.
\newblock {\em arXiv preprint arXiv:1509.02556}.

\bibitem[Rosenbaum, 1987]{rosenbaum1987sensitivity}
Rosenbaum, P.~R. (1987).
\newblock Sensitivity analysis for certain permutation inferences in matched observational studies.
\newblock {\em Biometrika}, 74(1):13--26.

\bibitem[Rosenbaum, 2004]{rosenbaum2004design}
Rosenbaum, P.~R. (2004).
\newblock Design sensitivity in observational studies.
\newblock {\em Biometrika}, 91(1):153--164.

\bibitem[Rosenbaum, 2010]{rosenbaum2010design}
Rosenbaum, P.~R. (2010).
\newblock Design sensitivity and efficiency in observational studies.
\newblock {\em Journal of the American Statistical Association}, 105(490):692--702.

\bibitem[Tan, 2006]{tan2006distributional}
Tan, Z. (2006).
\newblock A distributional approach for causal inference using propensity scores.
\newblock {\em Journal of the American Statistical Association}, 101(476):1619--1637.

\bibitem[Xu et~al., 2020]{xu2020applications}
Xu, L., Gotwalt, C., Hong, Y., King, C.~B., and Meeker, W.~Q. (2020).
\newblock Applications of the fractional-random-weight bootstrap.
\newblock {\em The American Statistician}, 74(4):345--358.

\bibitem[Zhao et~al., 2019]{zhao2019sensitivity}
Zhao, Q., Small, D.~S., and Bhattacharya, B.~B. (2019).
\newblock Sensitivity analysis for inverse probability weighting estimators via the percentile bootstrap.
\newblock {\em Journal of the Royal Statistical Society Series B: Statistical Methodology}, 81(4):735--761.

\end{thebibliography}

\clearpage
\appendix

\appendix

\renewcommand\thefigure{\thesection.\arabic{figure}}

\setcounter{figure}{0}

\section{Additional Discussion}\label{sec-app-discussion}

\subsection{Covariates and identifying assumption plausibility in the
running example}\label{sec-app-covar}

\citet{hk2022poland} use the following covariates in their analysis:

\begin{itemize}
\tightlist
\item
  indicators for the three subregions of Upper Silesia;
\item
  a wealth index calculated as an average of the number of shops,
  restaurants, and cinemas in each municipality;
\item
  the number of schools in each municipality, which the authors argue is
  a measure of state capacity;
\item
  an industrialization index built as an average of the presence of
  other mineral deposits and size of coal deposits, which the authors
  note is also correlated with poor working conditions and more
  grievances against the government;
\item
  an ethnic diversity index measured by the number of migrants into
  Upper Silesia after World War II;
\item
  the population of each municipality;\footnote{Unlike the original
    authors, we took the logarithm of this variable and also used the
    square root of population to normalize the counts of different
    buildings and establishments. In general, we find the original
    analysis highly sensitive to treatment of municipality populations.
    We do not further pursue these concerns in the present analysis,
    however, since they are not directly related to the causal
    assumptions under study.} and
\item
  an indicator for which municipalities were formerly under Russian
  occupation.
\end{itemize}

While these covariates capture several plausible causes of both
surveillance and antiregime sentiment, they are far from perfect. Many
of the covariates were justified by Hager and Krakowski as proxies for
broader underlying causes, such as ``state capacity;'' if the error in
the observed proxies (e.g., the number of schools) is at all correlated
with treatment and outcome, then the SOO assumptions would be violated.
Moreover, it is easy to imagine that municipalities that experienced
above-expected antiregime activity in some years---that is, more than
average given their covariates---might then be subjected to increased
surveillance. Since the variables here are aggregated over multiple
years, some of these dynamic effects could also lead to violations of
the SOO assumptions. We stress that the original authors conducted
extensive investigations into all of their identifying assumptions, and
collected additional data to test certain assumptions, and we lack the
space here to fully summarize their justifications. However, the threats
to causal inference under SOO were plausible enough that the authors
used an instrumental variable strategy for their primary analysis.

For the instrumental variable analysis, the authors point to several
facts to justify the exogeneity assumption. First, new priests were sent
to parishes upon the retirement or death of previous priests, which are
plausibly exogenous. Second, the Catholic Church excommunicated priests
found to have been corrupted, making it less likely that a priest's
corruptibility increased their chances of assignment. To justify the
exclusion restriction, Hager and Krakowski note the inherent secrecy of
the corrupted priests---the citizens protesting would not have had
knowledge of which priests were corrupted, so a backlash effect is less
likely. They also point to the roughly three-decade lag between the
collection of their corruption data and the protests. Finally, they
conduct several regression analyses on auxiliary data to argue against
the presence of several disallowed causal pathways.

However, it is not difficult to imagine plausible mechanisms that would
violate either assumption. Suppose the Catholic Church reserved
assignments to important cities, such as those with a history of protest
or surveillance, for priests that they judged to be more resistant to
recruitment. Alternatively, suppose the SB stepped up priest recruiting
efforts in cities they deemed more likely to protest or resist. Even
municipality population could be a confounder, since the count of
corrupted priests was not normalized by population, and population
enters the estimation equations linearly only.

For the exclusion restriction, the very lag between the collection of
the corruption data and the protests also allowed time for priests'
corruption to be exposed, which could affect protest in either direction
(instilling fear or provoking backlash). Hager and Krakowski also raise
the possibility of indoctrination of parishioners by the corrupted
priests.

\subsection{Proximal validity}\label{proximal-validity}

\begin{proposition}[Proximal
validity]\protect\hypertarget{prp-prox-valid}{}\label{prp-prox-valid}

Define \(\tau_{\prox}\) as the solution to the following two-stage least
squares procedure:

\begin{enumerate}
\def\labelenumi{\arabic{enumi}.}
\tightlist
\item
  Regress \(W_Y\) on \(Z\) and \(W_Z\) and generate predicted values
  \(\hat W_Y\).
\item
  Regress \(Y\) on \(Z\) and \(\hat W_Y\). The estimated coefficient on
  \(Z\) is denoted \(\tau_\prox\).
\end{enumerate}

\noindent Then, under Assumptions \ref{asm-prox-y} and \ref{asm-prox-z},
\(\tau_\prox = \tau\).

\end{proposition}

The proof, whose details appear in Appendix~\ref{sec-app-proofs}, relies
on the Frisch-Waugh-Lovell theorem: \begin{align*}
\tau_{\prox} &= \frac{\cov(Y^{\bot \hat W_Y}, Z^{\bot \hat W_Y})}{\var(Z^{\bot \hat W_Y})} \\
&= \frac{\cov\qty((\beta_u U + \beta_{w_y} W_Y + \beta_{w_z} W_Z + \tau Z + \eps_y)^{\bot \hat W_Y}, Z^{\bot \hat W_Y})}{\var(Z^{\bot \hat W_Y})} \\
&= \tau +
    \beta_u \frac{\cov(U^{\bot \hat W_Y}, Z^{\bot \hat W_Y})}{\var(Z^{\bot \hat W_Y})} +
    \beta_{w_y} \frac{\cov(W_Y^{\bot \hat W_Y}, Z^{\bot \hat W_Y})}{\var(Z^{\bot \hat W_Y})} +
    \beta_{w_z} \frac{\cov(W_Z^{\bot \hat W_Y}, Z^{\bot \hat W_Y})}{\var(Z^{\bot \hat W_Y})}.
\end{align*} The construction of \(\hat W_Y\) ensures
\(\cov(W_Y^{\bot \hat W_Y}, Z^{\bot \hat W_Y})=0\) (see
Lemma~\ref{lem-wy-y}), which also makes sense if \(\hat W_Y\) is viewed
as a proxy for \(U\). Assumption~\ref{asm-prox-z} guarantees
\(\beta_{w_z}=0\), and it turns out that Assumption~\ref{asm-prox-y}
guarantees that \(\cov(U^{\bot \hat W_Y}, Z^{\bot \hat W_Y})=0\), so the
two-stage approach is valid.

\subsection{\texorpdfstring{Partial \(R^2\) representation of
bias}{Partial R\^{}2 representation of bias}}\label{sec-app-r2}

In the main manuscript, we focus on the coefficient representations of
bias. While this allows for a simple and intuitive representation of
bias in terms of the underlying structural models, it can be difficult
to provide plausibility arguments about the magnitude of the different
coefficients.

Following \citet{cinelli2020making}, we can re-write the bias for each
identification strategy in terms of partial \(R^2\) values. To start, we
provide the bias of \(\tau_\soo\) in terms of partial \(R^2\) values,
which was first presented in \citet{cinelli2020making}:
\[\text{Bias}(\tau_\soo) = \frac{R_{Y\sim U\mid Z,W_Z, W_Y} \cdot R_{Z \sim U \mid W_Z,W_Y}}{
    \sqrt{1-R^2_{Z \sim U \mid W_Z,W_Y}}} \cdot
    \frac{\sd(Y^{\bot Z,W_Z, W_Y})}{\sd(Z^{\bot W_Z,W_Y})}.\]

Following the same approach of transforming the regression coefficients
into partial \(R^2\) values, we can similarly re-write the bias of
\(\tau_\iv\) and \(\tau_\prox\) in terms of partial \(R^2\) values,
formalized in the following corollaries.

\begin{corollary}[]\protect\hypertarget{cor-iv-r2}{}\label{cor-iv-r2}

The bias of \(\tau_\iv\) can be re-expressed in terms of partial \(R^2\)
values as: \begin{align*}
\text{Bias}(\tau_\iv) =& \frac{1}{R_{W_Z \sim Z \mid W_Y}} \cdot \frac{\sd(Y^{\bot W_Y, Z})}{\sd(Z^{\bot W_Y})} \cdot \left\{ \hliv{R_{Y \sim W_Z \mid W_Y, Z, U}} \sqrt{\frac{1-R^2_{Y \sim U \mid W_Y, Z}}{1-R^2_{W_Z \sim U \mid W_Y, Z}} \cdot \frac{1}{1-R^2_{Z \sim W_Z \mid W_Y}}} \right. \\
&\qquad \qquad +\left. \hlsoo{R_{Y\sim U \mid W_Z, W_Y, Z}} \cdot \hliv{R_{W_Z \sim U \mid W_Y}}\sqrt{\frac{1- R^2_{Y \sim W_Z \mid W_Y, Z}}{(1-R^2_{Z \sim U \mid W_Y, W_Z}) \cdot (1-R^2_{W_Z \sim U \mid W_Y})}} \right\} 
\end{align*} where the highlighted terms directly correspond to the
scale-invariant analog of the coefficients in the bias expression for
\(\tau_\iv\) in Theorem~\ref{thm-iv}.

\end{corollary}

\begin{corollary}[]\protect\hypertarget{cor-prox-r2}{}\label{cor-prox-r2}

The bias of \(\tau_\prox\) can be re-expressed in terms of partial
\(R^2\) values as: \begin{align*}
\bias(\tau_\prox) =&  S_{W_Z, Z}  \cdot \hliv{R_{Y \sim W_Z \mid W_Y, Z, U}} \sqrt{\frac{\var(Y^{\bot W_Y, Z})}{\var(Z^{\bot \hat W_Y})} \cdot \frac{\var(W_Z^{\bot \hat W_Y})}{\var(W_Z^{\bot W_Y})} \cdot \frac{(1-R^2_{Y \sim U \mid W_Y, Z})}{(1-R^2_{W_Z \sim U \mid W_Y, Z})(1-R^2_{Z \sim W_Z \mid W_Y})}} \\
&+ \hlsoo{R_{Y \sim U \mid W_Z, W_Y, Z}} \sqrt{\frac{\var(Y)}{\var(U)} \cdot \frac{1-R^2_{Y \sim W_Y, W_Z, Z}}{(1-R^2_{W_Z \sim U \mid W_Y, Z})(1-R^2_{U \sim W_Y, Z})}} \\
& \times \left( \hlprox{R_{W_Y \sim W_Z \mid U}} \cdot \hlsoo{R_{Z \sim U \mid W_Y, W_Z}} \frac{\sd(W_Z^{\bot U})}{\sd(W_Y^{\bot U})} \cdot  \frac{\sd(Z^{\bot W_Y, W_Z})}{\sd(U^{\bot W_Y, W_Z})} - \hlprox{R_{W_Y \sim Z \mid W_Z, U}}  \hliv{R_{W_Z \sim U \mid W_Y}} \frac{\sd(Z^{\bot U, W_Z})}{\sd(W_Y^{\bot U, W_Z})} \frac{\sd(W_Z^{\bot W_Y})}{\sd(U^{\bot W_Y})} \right) \\
& \times \frac{ \left( 1- R^2_{W_Y \sim U} \right) \var(W_Y) \var(U)}{\cov(W_Y, W_Z) \var(Z) - \cov(W_Y, Z) \cov(W_Z, Z)}
\end{align*} where the highlighted terms correspond to the
scale-invariant analog of the coefficients in the bias expression for
\(\tau_\prox\) in Theorem~\ref{thm-prox-bias}.

\end{corollary}

Unlike the partial \(R^2\) representation for \(\tau_\soo\), the bias
expressions for \(\tau_\iv\) and \(\tau_\prox\) in terms of partial
\(R^2\) values are more complicated, and involve several additional
partial \(R^2\) values that do not have a direct connection to the IV
and proximal assumptions. Notably, the parameterization is not variation
independent, which complicates interpretation. This motivates the
\(\rho\) parameterization used in Section 4 of the main manuscript,
which does allow for variation independent evaluation of the bias in
terms of four individual parameters. Furthermore, we can use the
\(\rho\) parameterization to back out the scale-invariant partial
\(R^2\) values that correspond to the coefficients in the bias
expressions for \(\tau_\iv\) and \(\tau_\prox\).

\subsection{IV bias under exogenous instruments}\label{sec-app-exog}

In some settings, researchers are able to exploit specific features of a
given study to ensure the instrument is exogenous by design (i.e.,
lotteries, random assignment to judges). When there is \emph{no}
violation in exogeneity of the instrument, we can simplify the bias
expression for \(\tau_{\iv}\) to only depend on the instrument strength
and violations in the exclusion restriction.

\begin{corollary}[Relative bias with exogenous
instruments]\protect\hypertarget{cor-relative-bias-iv}{}\label{cor-relative-bias-iv}

In settings where the instrument is exogenous, the bias of an IV
estimator simplifies to: \begin{align*}
\frac{1}{R_{Z \sim W_Z \mid W_Y}} \cdot \frac{\sd(Y^{\bot W_Y, Z})}{\sd(Z^{\bot W_Y})} \times R_{Y \sim W_Z \mid W_Y, Z, U} \cdot \sqrt{\left( 1-R^2_{Y \sim U \mid W_Y, Z}\right) \cdot \frac{1-R^2_{Z \sim U \mid W_Z, W_Y} (1-R^2_{Z \sim W_Z \mid W_Y})}{(1-R^2_{Z \sim W_Z \mid W_Y}) \cdot (1- R^2_{Z \sim U \mid W_Z, W_Y})}}
\end{align*} The relative bias between an IV estimator and an SOO
estimator will then depend on the degree to which the exclusion
restriction is violated. In general, when the following holds: \[
R^2_{Y \sim W_Z \mid W_Y, Z, U} \leq \frac{R^2_{Y \sim U \mid W_Y, W_Z, Z}}{1- \gamma \cdot R^2_{Y \sim U \mid W_Y,W_Z, Z}} \cdot \frac{R^2_{Z \sim U \mid W_Y, W_Z}}{1-R^2_{Z \sim U \mid W_Z, W_Y} (1-R^2_{Z \sim W_Z \mid W_Y})} \cdot (1-R^2_{Y \sim W_Z \mid W_Y, Z}) \cdot R^2_{Z \sim W_Z \mid W_Y},
\] then \(\tau_{\iv}\) will be less biased than \(\tau_\soo\), where
\(\gamma :=  R^2_{Y \sim U \mid W_Y, Z}/R^2_{Y \sim U \mid W_Y, W_Z, Z}\)
represents how different \(R^2_{Y \sim U \mid W_Y, W_Z, Z}\) is to
\(R^2_{Y \sim U \mid W_Y, Z}\).

\end{corollary}

For a fixed \(\gamma\) value, we can vary the selection-on-observables
parameters \(\sooRsq\) to find the maximum
\(R^2_{Y \sim W_Z \mid W_Y, Z, U}\) value that would still result in a
less biased estimator for \(\tau_\iv\) over \(\tau_\soo\). There exists
a set of \(\sooRsq\) that correspond to an upper bound on
\(R^2_{Y \sim W_Z \mid W_Y, Z, U}\) greater than 1. In that setting,
\(\tau_{\iv}\) will be less biased than \(\tau_\soo\). The bound on
\(R^2_{Y \sim W_Z \mid W_Y, Z, U}\) will increase as
\(R^2_{Y \sim U \mid W_Y, W_Z, Z}\) and \(R^2_{Z \sim U \mid W_Y, W_Z}\)
increase. Intuitively, this is because as \(\sooRsq\) increases in
magnitude, this implies that there is a greater amount of bias in
\(\tau_\soo\) from a stronger omitted confounder. As such, even under
relatively large violations of the exclusion restriction, \(\tau_{\iv}\)
may result in less biased estimates.

Notably, Corollary~\ref{cor-relative-bias-iv} provides a simple way to
compare the relative bias between \(\tau_\soo\) and \(\tau_\iv\).
However, this requires a relatively strong assumption that the
instrument used is exogenous. In settings when the instrumental variable
is not exogenous by design, researchers must simultaneously reason about
both violations in the exclusion restriction and exogeneity. This is
less straightforward, because as Theorem~\ref{thm-iv} highlights, the
bias from violations of exogeneity will interact with the
selection-on-observables bias.

\section{Additional sensitivity analysis and
benchmarking}\label{sec-app-sens}

\subsection{Detailed robustness value information for individual protest
outcome}\label{detailed-robustness-value-information-for-individual-protest-outcome}

Table~\ref{tbl-rv-detail} presents additional information on the
minimizing value of \(\rho\) in Table~\ref{tbl-rv}. We used a very small
penalty on the total magnitude of \(\rho\) in order to encourage a
unique solution to the optimization problem, but due to the limitations
of numerical optimization, these values of \(\rho\) should be
interpreted with some caution.

\begin{longtable}[]{@{}lrrlr@{}}
\caption[]{Robustness values and robustness allocations for the
reanalysis of \citet{hk2022poland}, individual protest outcome. The
median absolute deviations are shown in
parentheses.}\label{tbl-rv-detail}\tabularnewline
\toprule\noalign{}
Strategy & \(b\) & \(\TRV(b)\) & (std. err.) & Minimizing value of
\(\rho\) \\
\midrule\noalign{}
\endfirsthead
\toprule\noalign{}
Strategy & \(b\) & \(\TRV(b)\) & (std. err.) & Minimizing value of
\(\rho\) \\
\midrule\noalign{}
\endhead
\bottomrule\noalign{}
\endlastfoot
\textbf{SOO} & \(-0.296\) & \(0.331\) & \(\textcolor{gray}{(0.0902)}\) &
\((0.0001, 0.0001, 0.5756, -0.5756)\) \\
\textbf{IV} & \(-0.805\) & \(0.0173\) &
\(\textcolor{gray}{(1.43\times 10^{-4})}\) &
\((0.0001, -0.1315, 0.9900, -0.0577)\) \\
\textbf{Proximal} & \(-0.627\) & \(0.011\) &
\(\textcolor{gray}{(0.00358)}\) &
\((-0.0036, -0.4653, 0.9900, -0.0577)\) \\
\end{longtable}

\subsection{Detailed benchmarking values for individual protest
outcome}\label{detailed-benchmarking-values-for-individual-protest-outcome}

\footnotesize

{

\begin{longtable}[]{@{}lrrrrrr@{}}

\caption{\label{tbl-bench-ext}Additional benchmarking values for
individual protest outcome. Shown are the partial correlations measuring
each identifying assumption's violation and the benchmarked maximum
assumption violations. Median absolute deviations for each quantity are
shown in parentheses.}

\tabularnewline

\toprule\noalign{}
Covariate & \(\hlsoo{R_{Z\sim U\mid W_Z,W_Y}}\) &
\(\hlsoo{R_{Y\sim U\mid Z,W_Z,W_Y}}\) & \(\hliv{R_{W_Z\sim U\mid W_Y}}\)
& \(\hliv{R_{Y\sim W_Z\mid Z,W_Y,U}}\) &
\(\hlprox{R_{W_Y\sim Z\mid W_Z,U}}\) &
\(\hlprox{R_{W_Y\sim W_Z\mid U}}\) \\
\midrule\noalign{}
\endhead
\bottomrule\noalign{}
\endlastfoot
Wealth & \(0.0075~~\textcolor{gray}{(0.05)}\) &
\(0.044~~\textcolor{gray}{(0.04)}\) &
\(0.087~~\textcolor{gray}{(0.03)}\) &
\(-0.30~~\textcolor{gray}{(0.06)}\) &
\(-0.016~~\textcolor{gray}{(0.2)}\) &
\(0.01113~~\textcolor{gray}{(0.07)}\) \\
Schools & \(-0.0185~~\textcolor{gray}{(0.02)}\) &
\(0.043~~\textcolor{gray}{(0.04)}\) &
\(0.045~~\textcolor{gray}{(0.04)}\) &
\(-0.30~~\textcolor{gray}{(0.06)}\) &
\(-0.014~~\textcolor{gray}{(0.2)}\) &
\(-4.3\times 10^{-4}~~\textcolor{gray}{(0.07)}\) \\
Industrialization & \(-0.0413~~\textcolor{gray}{(0.03)}\) &
\(0.028~~\textcolor{gray}{(0.05)}\) &
\(0.013~~\textcolor{gray}{(0.06)}\) &
\(-0.30~~\textcolor{gray}{(0.06)}\) &
\(-0.018~~\textcolor{gray}{(0.2)}\) &
\(0.00514~~\textcolor{gray}{(0.07)}\) \\
Log population & \(0.0378~~\textcolor{gray}{(0.05)}\) &
\(-0.212~~\textcolor{gray}{(0.05)}\) &
\(0.038~~\textcolor{gray}{(0.03)}\) &
\(-0.31~~\textcolor{gray}{(0.06)}\) &
\(-0.017~~\textcolor{gray}{(0.2)}\) &
\(0.00438~~\textcolor{gray}{(0.07)}\) \\
Diversity & \(-0.1257~~\textcolor{gray}{(0.05)}\) &
\(0.057~~\textcolor{gray}{(0.05)}\) &
\(0.092~~\textcolor{gray}{(0.07)}\) &
\(-0.31~~\textcolor{gray}{(0.06)}\) &
\(-0.021~~\textcolor{gray}{(0.2)}\) &
\(0.00753~~\textcolor{gray}{(0.07)}\) \\
Katowickie region & \(-0.0130~~\textcolor{gray}{(0.04)}\) &
\(0.199~~\textcolor{gray}{(0.05)}\) &
\(0.051~~\textcolor{gray}{(0.04)}\) &
\(-0.32~~\textcolor{gray}{(0.06)}\) &
\(-0.015~~\textcolor{gray}{(0.2)}\) &
\(-0.00139~~\textcolor{gray}{(0.07)}\) \\
Opolskie region & \(0.0303~~\textcolor{gray}{(0.03)}\) &
\(0.101~~\textcolor{gray}{(0.04)}\) &
\(-0.037~~\textcolor{gray}{(0.04)}\) &
\(-0.30~~\textcolor{gray}{(0.06)}\) &
\(-0.016~~\textcolor{gray}{(0.2)}\) &
\(0.00379~~\textcolor{gray}{(0.07)}\) \\

\end{longtable}

}

\normalsize

\subsection{Sensitivity analysis and benchmarking for group protest
outcome}\label{sensitivity-analysis-and-benchmarking-for-group-protest-outcome}

In the following section, we provide an illustration of our sensitivity
analysis for the group protest outcome. In the main manuscript, we
focused on individual protest as an outcome. In this setting, depending
on which identification strategy was used, the resulting estimates
differed substantively. In the setting of the group protest outcome, all
three approaches yield relatively similar point estimates.

\begin{table}[h]
\centering
\begin{tabular}{>{\raggedright\arraybackslash}p{2.7cm}p{4cm}p{4cm}p{4cm}} \toprule
\textbf{} & \textbf{SOO} & \textbf{IV}  & \textbf{Proximal} \\
\midrule
 Estimate (s.e.) for group protest
  & $\phantom{-}0.347\ \ (0.058)$
  & $\phantom{-}0.347\ \ (0.128)$
  & $\phantom{-}0.348\ \ (0.373)$ \\ \bottomrule 
\end{tabular} 
\caption{Estimates for group protest outcome, across all three identification approaches.}
\end{table}

Table~\ref{tbl-rv-grp}, Table~\ref{tbl-bench-grp}, and
Table~\ref{tbl-bench-grp-ext} show the generalized robustness value and
benchmarking results for the group protest outcome, paralleling the
results presented in Section~\ref{sec-rv-demo} and the preceding section
for the individual protest outcome.

\setstretch{1.0}

\begin{longtable}[]{@{}lrrll@{}}
\caption[]{Generalized robustness values and robustness allocations for
the reanalysis of \citet{hk2022poland}, group protest outcome. Median
absolute deviations are shown in
parentheses.}\label{tbl-rv-grp}\tabularnewline
\toprule\noalign{}
Strategy & \(b\) & \(\TRV(b)\) & (std. err.) & Minimizing value of
\(\rho\) \\
\midrule\noalign{}
\endfirsthead
\toprule\noalign{}
Strategy & \(b\) & \(\TRV(b)\) & (std. err.) & Minimizing value of
\(\rho\) \\
\midrule\noalign{}
\endhead
\bottomrule\noalign{}
\endlastfoot
\textbf{SOO} & \(0.347\) & \(0.333\) & \(\textcolor{gray}{(0.0715)}\) &
\((0.0001, 0.0001, 0.5767, 0.5767)\) \\
\textbf{IV} & \(0.347\) & \(0.00304\) &
\(\textcolor{gray}{(3.38\times 10^{-5})}\) &
\((0.0001, 0.0001, 0.9900, 0.0580)\) \\
\textbf{Proximal} & \(0.348\) & \(0.00304\) &
\(\textcolor{gray}{(0.0017)}\) &
\((0.0260, 0.0000, -0.9900, -0.0580)\) \\
\end{longtable}

{

\begin{longtable}[]{@{}lrrrr@{}}

\caption{\label{tbl-bench-grp}Benchmarking values for group protest
outcome. Shown are the benchmarked SOO parameters and the benchmarked
maximum assumption violations and the benchmarked maximum assumption
violations. Median absolute deviations for each quantity are shown in
parentheses.}

\tabularnewline

\toprule\noalign{}
Covariate & Change in \(\tau\) &
\(\norm{A_\soo(\hat \rho(j))}_\infty^2\) &
\(\norm{A_\iv(\hat \rho(j))}_\infty^2\) &
\(\norm{A_\prox(\hat \rho(j))}_\infty^2\) \\
\midrule\noalign{}
\endhead
\bottomrule\noalign{}
\endlastfoot
Wealth & \(0.004342~~\textcolor{gray}{(0.012)}\) &
\(0.00550~~\textcolor{gray}{(0.008)}\) &
\(0.00560~~\textcolor{gray}{(0.006)}\) &
\(3.5\times 10^{-4}~~\textcolor{gray}{(0.1)}\) \\
Schools & \(2.2\times 10^{-5}~~\textcolor{gray}{(0.001)}\) &
\(1.6\times 10^{-4}~~\textcolor{gray}{(0.004)}\) &
\(0.00683~~\textcolor{gray}{(0.007)}\) &
\(3.2\times 10^{-4}~~\textcolor{gray}{(0.1)}\) \\
Industrialization & \(0.002716~~\textcolor{gray}{(0.003)}\) &
\(0.00302~~\textcolor{gray}{(0.004)}\) &
\(4.5\times 10^{-4}~~\textcolor{gray}{(0.006)}\) &
\(4.4\times 10^{-4}~~\textcolor{gray}{(0.1)}\) \\
Log population & \(0.050160~~\textcolor{gray}{(0.075)}\) &
\(0.04131~~\textcolor{gray}{(0.020)}\) &
\(0.00139~~\textcolor{gray}{(0.005)}\) &
\(3.4\times 10^{-4}~~\textcolor{gray}{(0.1)}\) \\
Diversity & \(0.010962~~\textcolor{gray}{(0.009)}\) &
\(0.01696~~\textcolor{gray}{(0.010)}\) &
\(0.00843~~\textcolor{gray}{(0.011)}\) &
\(5.1\times 10^{-4}~~\textcolor{gray}{(0.1)}\) \\
Katowickie region & \(0.013613~~\textcolor{gray}{(0.025)}\) &
\(0.01866~~\textcolor{gray}{(0.016)}\) &
\(0.00245~~\textcolor{gray}{(0.005)}\) &
\(3.1\times 10^{-4}~~\textcolor{gray}{(0.1)}\) \\
Opolskie region & \(7.2\times 10^{-4}~~\textcolor{gray}{(0.002)}\) &
\(5.3\times 10^{-4}~~\textcolor{gray}{(0.003)}\) &
\(0.00206~~\textcolor{gray}{(0.005)}\) &
\(3.2\times 10^{-4}~~\textcolor{gray}{(0.1)}\) \\

\end{longtable}

}

\footnotesize

{

\begin{longtable}[]{@{}lrrrrrr@{}}

\caption{\label{tbl-bench-grp-ext}Additional benchmarking values for
group protest outcome. Shown are the partial correlations measuring each
identifying assumption's violation and the benchmarked maximum
assumption violations. Median absolute deviations for each quantity are
shown in parentheses.}

\tabularnewline

\toprule\noalign{}
Covariate & \(\hlsoo{R_{Z\sim U\mid W_Z,W_Y}}\) &
\(\hlsoo{R_{Y\sim U\mid Z,W_Z,W_Y}}\) & \(\hliv{R_{W_Z\sim U\mid W_Y}}\)
& \(\hliv{R_{Y\sim W_Z\mid Z,W_Y,U}}\) &
\(\hlprox{R_{W_Y\sim Z\mid W_Z,U}}\) &
\(\hlprox{R_{W_Y\sim W_Z\mid U}}\) \\
\midrule\noalign{}
\endhead
\bottomrule\noalign{}
\endlastfoot
Wealth & \(-0.0170~~\textcolor{gray}{(0.07)}\) &
\(-0.074~~\textcolor{gray}{(0.07)}\) &
\(0.075~~\textcolor{gray}{(0.03)}\) &
\(0.00529~~\textcolor{gray}{(0.05)}\) &
\(-0.019~~\textcolor{gray}{(0.2)}\) &
\(-0.0091~~\textcolor{gray}{(0.06)}\) \\
Schools & \(0.0016~~\textcolor{gray}{(0.03)}\) &
\(0.013~~\textcolor{gray}{(0.08)}\) &
\(0.083~~\textcolor{gray}{(0.05)}\) &
\(-0.00120~~\textcolor{gray}{(0.05)}\) &
\(-0.018~~\textcolor{gray}{(0.2)}\) &
\(-0.0151~~\textcolor{gray}{(0.06)}\) \\
Industrialization & \(-0.0401~~\textcolor{gray}{(0.03)}\) &
\(-0.055~~\textcolor{gray}{(0.04)}\) &
\(0.021~~\textcolor{gray}{(0.06)}\) &
\(0.00177~~\textcolor{gray}{(0.05)}\) &
\(-0.021~~\textcolor{gray}{(0.2)}\) &
\(-0.0106~~\textcolor{gray}{(0.07)}\) \\
Log population & \(0.0574~~\textcolor{gray}{(0.07)}\) &
\(0.203~~\textcolor{gray}{(0.05)}\) &
\(0.037~~\textcolor{gray}{(0.03)}\) &
\(-0.00193~~\textcolor{gray}{(0.05)}\) &
\(-0.018~~\textcolor{gray}{(0.2)}\) &
\(-0.0126~~\textcolor{gray}{(0.07)}\) \\
Diversity & \(-0.1302~~\textcolor{gray}{(0.05)}\) &
\(-0.073~~\textcolor{gray}{(0.05)}\) &
\(0.092~~\textcolor{gray}{(0.07)}\) &
\(0.00993~~\textcolor{gray}{(0.05)}\) &
\(-0.023~~\textcolor{gray}{(0.2)}\) &
\(-0.0090~~\textcolor{gray}{(0.07)}\) \\
Katowickie region & \(-0.0327~~\textcolor{gray}{(0.07)}\) &
\(-0.137~~\textcolor{gray}{(0.07)}\) &
\(0.049~~\textcolor{gray}{(0.04)}\) &
\(0.00785~~\textcolor{gray}{(0.05)}\) &
\(-0.014~~\textcolor{gray}{(0.1)}\) &
\(-0.0177~~\textcolor{gray}{(0.07)}\) \\
Opolskie region & \(0.0198~~\textcolor{gray}{(0.03)}\) &
\(0.023~~\textcolor{gray}{(0.05)}\) &
\(-0.045~~\textcolor{gray}{(0.04)}\) &
\(8.8\times 10^{-4}~~\textcolor{gray}{(0.05)}\) &
\(-0.018~~\textcolor{gray}{(0.2)}\) &
\(-0.0124~~\textcolor{gray}{(0.06)}\) \\

\end{longtable}

}

\normalsize

\setstretch{\papersp}

\newpage

\section{Mean-squared error contour plot}\label{sec-app-mse}

Figure~\ref{fig-contour-mse} is analogous to Figure~\ref{fig-contour}:
the plotted contours are the SOO bias. The contours are colored and
divided however, according to which estimator has the lowest
mean-squared error. Because in this case the three estimators have
similar variance, the plot is very similar to Figure~\ref{fig-contour}.

\begin{figure}

\centering{

\includegraphics[width=6.25in,height=\textheight,keepaspectratio]{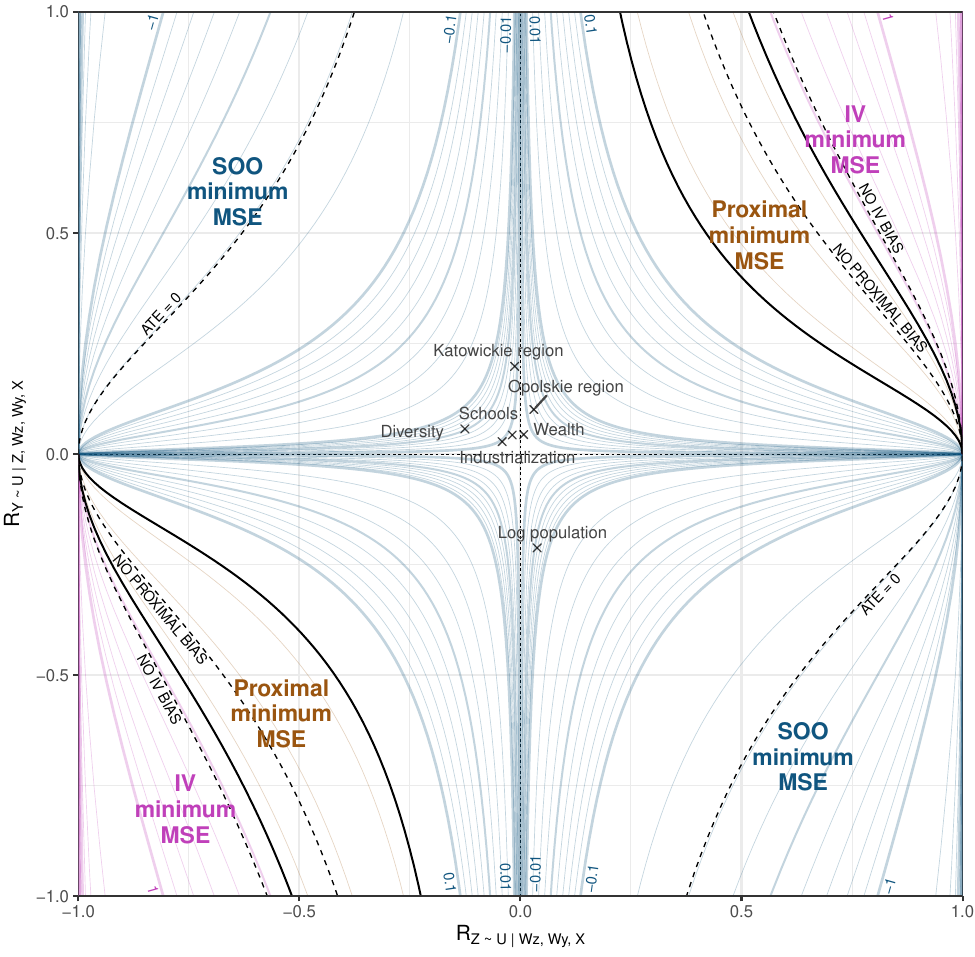}

}

\caption{\label{fig-contour-mse}Sensitivity plot for the running
example. The contour lines are on a logarithmic scale and show the
amount of SOO confounding bias. Contours are colored according to the
regions where each estimator has the lowest MSE. Gray crosses indicate
benchmark sensitivity parameter values for each of the covariates.}

\end{figure}%

\FloatBarrier

\section{Proofs}\label{sec-app-proofs}

For clarity in what follows, we abuse notation and manipulate random
variables as if they were random vectors of sufficient length \(N\), and
then implicitly send \(N\to\infty\) when converting back to statements
in terms of variances, covariances, etc.

\begin{lemma}[]\protect\hypertarget{lem-wy-y}{}\label{lem-wy-y}

\(\cov(W_Y^{\bot \hat W_Y}, Z^{\bot \hat W_Y})=0\).

\end{lemma}

\begin{proof}
Let \(\tilde W_Z = \mqty[Z & W_Z]\), and \(P_{\tilde W_Z}\) and
\(P_{\hat W_Y}\) be the projection matrices onto \(\tilde W_Z\) and
\(\hat W_Y\), respectively. Then \begin{align*}
P_{\hat W_Y} W_Y &= P_{\tilde W_Z} W_Y (W_Y^\top P_{\tilde W_Z}^\top P_{\tilde W_Z} W_Y)^{-1} W_Y^\top P_{\tilde W_Z}^\top W_Y \\
&= P_{\tilde W_Z} W_Y = \hat W_Y.
\end{align*} Then \begin{align*}
\cov(W_Y^{\bot \hat W_Y}, Z^{\bot \hat W_Y}) &= \qty{(I- P_{\hat W_Y}) W_Y}^\top \qty{(I- P_{\hat W_Y}) Z} \\
&= W_Y^\top (I- P_{\hat W_Y})^\top Z  \\
&= (W_Y - \hat W_Y)^\top Z \\
&= 0,
\end{align*} since least-squares residuals are orthogonal to the
predictors (of which \(Z\) is one).
\end{proof}

\begin{lemma}[]\protect\hypertarget{lem-soo}{}\label{lem-soo}

Let \(\tau_\soo\) be the estimated coefficient in front of \(Z\) with
the regression specification \(Y \sim Z + W_Z + W_Y\). Then, the bias
that arises from omitting \(U\) is: \[
\bias(\tau_\soo)
= \beta_u \frac{\cov(U^{\bot W_Z,W_Y}, Z^{\bot W_Z,W_Y})}{\var(Z^{\bot W_Z,W_Y})}
= \frac{R_{Y\sim U\mid Z,W_Z, W_Y} \cdot R_{Z \sim U \mid W_Z,W_Y}}{
    \sqrt{1-R^2_{Z \sim U \mid W_Z,W_Y}}} \cdot
    \frac{\sd(Y^{\bot Z,W_Z, W_Y})}{\sd(Z^{\bot W_Z,W_Y})}
\]

\end{lemma}

\begin{proof}
This result follows directly from \citet{cinelli2020making}. We provide
the full derivation for reference: \begin{align*}
\bias(\hat\tau_\soo) &=  \beta_u \frac{\cov(U^{\bot W_Z, W_Y}, Z^{\bot W_Z, W_Y})}{\var(Z^{\bot W_Z, W_Y})}\\
&= \frac{\cov(U^{\bot W_Z,W_Y}, Z^{\bot W_Z,W_Y})}{\var(Z^{\bot W_Z,W_Y})}
\cdot \frac{\cov(Y^{\bot Z,W_Z, W_Y}, U^{\bot Z, W_Z,W_Y})}{\var(U^{\bot Z,W_Z, W_Y})} \\
&= \frac{\cor(U^{\bot W_Z,W_Y}, Z^{\bot W_Z,W_Y}) \cdot
\sd(U^{\bot W_Z,W_Y})}{\sd(Z^{\bot W_Z,W_Y})} \cdot
\frac{\cor(Y^{\bot Z, W_Z,W_Y}, U^{\bot Z, W_Z,W_Y}) \cdot
\sd(Y^{\bot Z,W_Z, W_Y})}{\sd(U^{\bot Z,W_Z, W_Y})}\\
&= R_{Z\sim U\mid W_Z,W_Y}R_{Y\sim U\mid Z, W_Z,W_Y} \cdot \frac{\sd(U^{\bot W_Z, W_Y})}{\sd(U^{\bot Z,W_Z, W_Y})} \cdot \frac{\sd(Y^{\bot Z,W_Z, W_Y})}{\sd(Z^{\bot W_Z, W_Y})}\\
&= \frac{R_{Y\sim U\mid Z,W_Z, W_Y} \cdot R_{Z \sim U \mid W_Z,W_Y}}{
    \sqrt{1-R^2_{Z \sim U \mid W_Z,W_Y}}} \cdot
    \frac{\sd(Y^{\bot Z,W_Z, W_Y})}{\sd(Z^{\bot W_Z,W_Y})}
\end{align*}
\end{proof}

\subsection{\texorpdfstring{Proposition~\ref{prp-prox-valid}}{Proposition~}}\label{proposition-prp-prox-valid}

\begin{proof}
As in the main text, use Frisch-Waugh-Lovell to write
\(\tau_{\prox}-\tau\) as \begin{align*}
\tau_{\prox} -\tau &=
    \beta_u \frac{\cov(U^{\bot \hat W_Y}, Z^{\bot \hat W_Y})}{\var(Z^{\bot \hat W_Y})} +
    \beta_{w_y} \frac{\cov(W_Y^{\bot \hat W_Y}, Z^{\bot \hat W_Y})}{\var(Z^{\bot \hat W_Y})} +
    \beta_{w_z} \frac{\cov(W_Z^{\bot \hat W_Y}, Z^{\bot \hat W_Y})}{\var(Z^{\bot \hat W_Y})}.
\end{align*}

The second covariance is zero by Lemma~\ref{lem-wy-y}. Moreover,
Assumption~\ref{asm-prox-z} is equivalent to \(\beta_{w_z}=0\), so the
third term is zero.

We next show \(\cov(U^{\bot \hat W_Y}, Z^{\bot \hat W_Y})=0\). First,
notice that \begin{equation}
\cov(W_Z,\eps_{w_y})=\varphi_u\cov(U,\eps_{w_y})+\varphi_{w_y}\cov(W_Y,\eps_{w_y}) + \cov(\eps_{w_z}, \eps_{w_y})=\varphi_{w_y}\var(\eps_{w_y}).
\end{equation} Under Assumption~\ref{asm-prox-y}, \(\varphi_{w_y} = 0\),
so \(\cov(W_Z, \eps_{w_y})=0\). Similarly, \begin{equation}
\cov(Z,\eps_{w_y})=\gamma_u\cov(U,\eps_{w_y})+\gamma_{w_y}\cov(W_Y,\eps_{w_y})+\gamma_{w_z}\cov(W_Z,\eps_{w_y}) + \cov(\eps_{z}, \eps_{w_y})
=(\gamma_{w_y} + \gamma_{w_z}\varphi_{w_y})\var(\eps_{w_y}).
\label{eqn:cov_z_eps_wy}
\end{equation}

Under Assumption~\ref{asm-prox-y}, \(\gamma_{w_y}=0\) and
\(\varphi_{w_y} = 0\), so \(\cov(Z, \eps_{w_y})=0\).

Then let \(\tilde W_Z = \mqty[Z & W_Z]\), so we can write \(\hat W_Y\)
as \[
\hat W_Y = \alpha_u P_{\tilde W_Z} U + P_{\tilde W_Z} \eps_{w_y} = \alpha_u P_{\tilde W_Z} U
\] where \(P_{\tilde W_Z}\) represents the projection matrix onto
\(\tilde W_Z\), and \(P_{\tilde W_Z} \eps_{w_y} = 0\) since we have
shown that the covariance between \(\eps_{w_y}\) and both \(W_Z\) and
\(Z\) are zero. Notice that \(P_{\tilde W_Z}Z=Z\), since \(Z\) is in the
column space of \(\tilde W_Z\) by construction. Thus we can write (with
the \(\alpha_u\) cancelling) \begin{align*}
P_{\hat W_Y} Z &= P_{\tilde W_Z} U \big(U^\top P_{\tilde W_Z}^\top P_{\tilde W_Z} U \big)^{-1} U^\top P_{\tilde W_Z}^\top Z \\
&= P_{\tilde W_Z} U \big(U^\top P_{\tilde W_Z} U \big)^{-1} U^\top Z.
\end{align*} Then: \begin{align*}
\cov(U^{\bot \hat W_Y}, Z^{\bot \hat W_Y}) &= \qty{ (I- P_{\hat W_Y}) U}^\top \qty{(I- P_{\hat W_Y}) Z} \\
&= U^\top (I- P_{\hat W_Y}) Z \\
&= U^\top Z - U^\top P_{\hat W_Y} Z \\
&= U^\top Z - U^\top P_{\tilde W_Z} U \big(U^\top P_{\tilde W_Z} U \big)^{-1} U^\top Z \\
&= 0. \qedhere
\end{align*}
\end{proof}

\subsection{\texorpdfstring{Theorem~\ref{thm-iv}}{Theorem~}}\label{theorem-thm-iv}

\begin{proof}
By Frisch-Waugh-Lovell and Lemma~\ref{lem-wy-y}, we can write the bias
in \(\tau_\iv\) as
\begin{equation}\protect\phantomsection\label{eq-bias-iv-1}{
\tau_\iv - \tau = \underbrace{\beta_{w_z} \cdot \frac{\cov(\hat Z^{\bot W_Y}, W_Z^{\bot W_Y})}{\var(\hat Z^{\bot W_Y})}}_{(1)} + \underbrace{\beta_u \frac{\cov(\hat Z^{\bot W_Y}, U^{\bot W_Y})}{\var(\hat Z^{\bot W_Y})}}_{(2)}.
}\end{equation}

Term 1 represents violations in the exclusion restriction. We can
re-write term 1 as follows, where \(\hat \gamma\) represent the
coefficient on \(W_Z\) in the regression of \(Z\) on \(W_Z\) and
\(W_Y\): \begin{align*}
\beta_{w_z} \cdot \frac{\cov(\hat Z^{\bot W_Y}, W_Z^{\bot W_Y})}{\var(\hat Z^{\bot W_Y})}
&= \beta_{w_z} \cdot \frac{\cov((\hat\gamma W_Z)^{\bot W_Y}, W_Z^{\bot W_Y})}{\var((\hat\gamma W_Z)^{\bot W_Y})} \\
&= \beta_{w_z} \cdot \frac{\hat \gamma  \cdot \cov(W_Z^{\bot W_Y}, W_Z^{\bot W_Y})}{\hat \gamma^2 \var(W_Z^{\bot W_Y})}\\
&= \beta_{w_z} \cdot \frac{1}{\hat \gamma} \\
&= \beta_{w_z} \cdot \frac{\var(W_Z^{\bot W_Y})}{\cov(W_Z^{\bot W_Y}, Z^{\bot W_Y})} \\
&= \beta_{w_z} \cdot \frac{1}{\cor(W_Z^{\bot W_Y}, Z^{\bot W_Y})} \cdot \frac{\sd(W_Z^{\bot W_Y})}{\sd(Z^{\bot W_Y})} \\
\end{align*}

Term 2 represents violations in instrument exogeneity: \begin{align*}
\beta_u \frac{\cov(\hat Z^{\bot W_Y}, U^{\bot W_Y})}{\var(\hat Z^{\bot W_Y})}
&= \beta_u \cdot \frac{1}{\hat \gamma} \cdot \frac{\cov( W_Z^{\bot W_Y}, U^{\bot W_Y})}{\var( W_Z^{\bot W_Y})} \\
&= \beta_u \cdot \frac{\var(W_Z^{\bot W_Y})}{\cov(W_Z^{\bot W_Y}, Z^{\bot W_Y})}\cdot \frac{\cov( W_Z^{\bot W_Y}, U^{\bot W_Y})}{\var( W_Z^{\bot W_Y})} \\
&= \beta_u \cdot \frac{\cov(W_Z^{\bot W_Y}, U^{\bot W_Y})}{\cov(W_Z^{\bot W_Y}, Z^{\bot W_Y})} \\
&= \beta_u \cdot \frac{\cov(W_Z^{\bot W_Y}, U^{\bot W_Y})}{\cov(W_Z^{\bot W_Y}, Z^{\bot W_Y})} \cdot \frac{\var(U^{\bot W_Y})}{\var(U^{\bot W_Y})}\\
&= \beta_u \cdot \frac{\varphi_u}{\cov(W_Z^{\bot W_Y}, Z^{\bot W_Y})} \cdot \var(U^{\bot W_Y}) \\
&= \beta_u \cdot \frac{\varphi_u}{\cor(W_Z^{\bot W_Y}, Z^{\bot W_Y})} \cdot \frac{\var(U^{\bot W_Y})}{\sd(W_Z^{\bot W_Y}) \sd(Z^{\bot W_Y})}
\end{align*} The correlation term (i.e.,
\(\cor(W_Z^{\bot W_Y}, U^{\bot W_Y})\)) will equal zero if IV exogeneity
holds.

\noindent Putting it together: \begin{align*}
\tau_\iv - \tau  =&\beta_{w_z} \cdot \frac{1}{\cor(W_Z^{\bot W_Y}, Z^{\bot W_Y})} \cdot \frac{\sd(W_Z^{\bot W_Y})}{\sd(Z^{\bot W_Y})} + \beta_u \cdot \frac{\varphi_u}{\cor(W_Z^{\bot W_Y}, Z^{\bot W_Y})} \cdot \frac{\var(U^{\bot W_Y})}{\sd(W_Z^{\bot W_Y}) \sd(Z^{\bot W_Y})}\\
=& \frac{1}{R_{W_Z \sim Z \mid W_Y}} \cdot \frac{\sd(W_Z^{\bot W_Y})}{\sd(Z^{\bot W_Y})}\left( \beta_{w_z} + \beta_u \cdot \varphi_u \cdot \frac{\var(U^{\bot W_Y})}{\var(W_Z^{\bot W_Y})}\right)
\end{align*}
\end{proof}

\subsection{\texorpdfstring{Corollary~\ref{cor-iv-r2}}{Corollary~}}\label{corollary-cor-iv-r2}

\begin{proof}
We can then re-write \(\beta_{w_z}\) and \(\beta_u\) in terms of \(R^2\)
values. Throughout, we will make use of the following decomposition for
partial \(R^2\) values: \begin{align*}
R^2_{Y \sim A \mid B, C} = \frac{R^2_{Y \sim A, B, C} - R^2_{Y \sim B, C}}{1-R^2_{Y \sim B, C}} \implies R^2_{Y \sim A, B, C} = R^2_{Y \sim A \mid B, C}(1-R^2_{Y \sim B, C}) + R^2_{Y \sim B, C}
\end{align*}

\begin{align}
\beta_{w_z} &= \frac{\cov(W_Z^{\bot W_Y, U, Z}, Y^{\bot W_Y, U, Z})}{\var(W_Z^{\bot W_Y, U, Z})} \nonumber \\
&= R_{W_Z \sim Y \mid W_Y, U, Z} \cdot \frac{\sd(Y^{\bot W_Y, U, Z})}{\sd(W_Z^{\bot W_Y, U, Z})} \nonumber \\
&= R_{Y \sim W_Z \mid W_Y, Z, U} \sqrt{\frac{\var(Y)}{\var(W_Z)} \cdot \frac{(1-R^2_{Y \sim U \mid W_Y, Z}) (1-R^2_{Y \sim W_Y, Z})}{(1-R^2_{W_Z \sim U \mid W_Y, Z})(1-R^2_{W_Z \sim W_Y, Z})}}
\label{eqn:beta_wz_r2_1}\\
&= R_{Y \sim W_Z \mid W_Y, Z, U} \sqrt{\frac{\var(Y^{\bot W_Y, Z})}{\var(W_Z^{\bot W_Y})} \cdot \frac{(1-R^2_{Y \sim U \mid W_Y, Z})}{(1-R^2_{W_Z \sim U \mid W_Y, Z})(1-R^2_{Z \sim W_Z \mid W_Y})}}
\label{eqn:beta_wz_r2_2}
\end{align}

We can re-write
\(\beta_u \cdot \varphi_u  \cdot \frac{\var(U^{\bot W_Y})}{\var(W_Z^{\bot W_Y})}\)
as: \begin{align*}
\beta_u &\cdot \varphi_u \cdot  \frac{\var(U^{\bot W_Y})}{\var(W_Z^{\bot W_Y})}\\
=& R_{Y\sim U \mid W_Z, W_Y, Z} \cdot R_{W_Z \sim U \mid W_Y} \cdot \frac{\sd(Y^{\bot W_Z, W_Y, Z})}{\sd(U^{\bot W_Z, W_Y, Z})} \cdot \frac{\sd(U^{\bot W_Y})}{\sd(W_Z^{\bot W_Y})} \\
=& R_{Y\sim U \mid W_Z, W_Y, Z} \cdot R_{W_Z \sim U \mid W_Y} \cdot\frac{\sd(U^{\bot W_Y})}{\sd(U^{\bot W_Z, W_Y, Z})}  \frac{\sd(Y^{\bot W_Z, W_Y, Z})}{\sd(W_Z^{\bot W_Y})} \\
=& R_{Y\sim U \mid W_Z, W_Y, Z} \cdot R_{W_Z \sim U \mid W_Y} \cdot \sqrt{\frac{1}{1-R^2_{U \sim W_Z, Z \mid W_Y}}}  \frac{\sd(Y^{\bot W_Z, W_Y, Z})}{\sd(W_Z^{\bot W_Y})}  \\
=& R_{Y\sim U \mid W_Z, W_Y, Z} \cdot R_{W_Z \sim U \mid W_Y} \\
& \times \sqrt{\frac{1}{(1-R^2_{U \sim Z \mid W_Z, W_Y})(1-R^2_{U \sim W_Z \mid W_Y})}}  \frac{\sd(Y^{\bot W_Z, W_Y, Z})}{\sd(W_Z^{\bot W_Y})}
\end{align*}

Combining together: \begin{align*}
\bias(\tau_{\iv}) =&\frac{1}{R_{W_Z \sim Z \mid W_Y}} \cdot \frac{\sd(W_Z^{\bot W_Y})}{\sd(Z^{\bot W_Y})} \left\{R_{Y \sim W_Z \mid W_Y, Z, U} \sqrt{\frac{\var(Y^{\bot W_Y, Z})}{\var(W_Z^{\bot W_Y})} \cdot \frac{(1-R^2_{Y \sim U \mid W_Y, Z})}{(1-R^2_{W_Z \sim U \mid W_Y, Z})(1-R^2_{Z \sim W_Z \mid W_Y})}}
\right. \\
&\quad +\left. R_{Y\sim U \mid W_Z, W_Y, Z} \cdot R_{W_Z \sim U \mid W_Y}\sqrt{\frac{1}{1-R^2_{W_Z \sim U \mid W_Y}} \cdot \frac{1}{1-R^2_{Z \sim U \mid W_Y, W_Z}}} \cdot \frac{\sd(Y^{\bot W_Z, W_Y, Z})}{\sd(W_Z^{\bot W_Y})}\right\} \\
=&\frac{1}{R_{W_Z \sim Z \mid W_Y}} \cdot \frac{\sd(Y^{\bot W_Y, Z})}{\sd(Z^{\bot W_Y})} \cdot \left\{ R_{Y \sim W_Z \mid W_Y, Z, U} \sqrt{\frac{1-R^2_{Y \sim U \mid W_Y, Z}}{1-R^2_{W_Z \sim U \mid W_Y, Z}} \cdot \frac{1}{1-R^2_{Z \sim W_Z \mid W_Y}}} \right. \\
&\quad +\left. R_{Y\sim U \mid W_Z, W_Y, Z} \cdot R_{W_Z \sim U \mid W_Y}\sqrt{\frac{1}{1-R^2_{W_Z \sim U \mid W_Y}} \cdot \frac{1}{1-R^2_{Z \sim U \mid W_Y, W_Z}}} \cdot \frac{\sd(Y^{\bot W_Z, W_Y, Z})}{\sd(Y^{\bot W_Y, Z})}\right\} \\
=&\frac{1}{R_{W_Z \sim Z \mid W_Y}} \cdot \frac{\sd(Y^{\bot W_Y, Z})}{\sd(Z^{\bot W_Y})} \cdot \left\{ R_{Y \sim W_Z \mid W_Y, Z, U} \sqrt{\frac{1-R^2_{Y \sim U \mid W_Y, Z}}{1-R^2_{W_Z \sim U \mid W_Y, Z}} \cdot \frac{1}{1-R^2_{Z \sim W_Z \mid W_Y}}} \right. \\
&\qquad +\left. R_{Y\sim U \mid W_Z, W_Y, Z} \cdot R_{W_Z \sim U \mid W_Y}\sqrt{\frac{1}{1-R^2_{W_Z \sim U \mid W_Y}} \cdot \frac{1-R^2_{Y \sim W_Z \mid W_Y, Z}}{1-R^2_{Z \sim U \mid W_Y, W_Z}}}\right\} \\
=&\frac{1}{R_{W_Z \sim Z \mid W_Y}} \cdot \frac{\sd(Y^{\bot W_Y, Z})}{\sd(Z^{\bot W_Y})} \cdot \left\{ R_{Y \sim W_Z \mid W_Y, Z, U} \sqrt{\frac{1-R^2_{Y \sim U \mid W_Y, Z}}{1-R^2_{W_Z \sim U \mid W_Y, Z}} \cdot \frac{1}{1-R^2_{Z \sim W_Z \mid W_Y}}} \right. \\
&\qquad \qquad +\left. R_{Y\sim U \mid W_Z, W_Y, Z} \cdot R_{W_Z \sim U \mid W_Y}\sqrt{\frac{1- R^2_{Y \sim W_Z \mid W_Y, Z}}{(1-R^2_{Z \sim U \mid W_Y, W_Z}) \cdot (1-R^2_{W_Z \sim U \mid W_Y})}} \right\}
\end{align*}
\end{proof}

\subsection{\texorpdfstring{Corollary~\ref{cor-relative-bias-iv}}{Corollary~}}\label{corollary-cor-relative-bias-iv}

\begin{proof}
To start, we apply the results from Corollary A.1, and set $R^2_{W_Z \sim U \mid W_Y} = 0$, which corresponds to the setting when exogeneity holds. The resulting bias is simply a function of the first term in the bias expression: 

\begin{equation} 
\frac{1}{R_{Z \sim W_Z \mid W_Y}} \cdot \frac{\sd(Y^{\bot W_Y, Z})}{\sd(Z^{\bot W_Y})} \times R_{Y \sim W_Z \mid W_Y, Z, U} \cdot \sqrt{\left( 1-R^2_{Y \sim U \mid W_Y, Z}\right) \cdot \frac{1}{(1-R^2_{Z \sim W_Z \mid W_Y}) \cdot (1-R^2_{W_Z \sim U \mid W_Y, Z})}}.
\label{eqn:iv-exog1}
\end{equation} 

We can re-write $1-R^2_{W_Z \sim U \mid W_Y, Z}$. Under exogeneity and by the first-order partial correlation formula, 
\begin{align*} 
R_{W_Z \sim U \mid W_Y, Z}& = \frac{R_{W_Z \sim U \mid W_Y} - R_{W_Z \sim Z \mid W_Y} R_{U \sim Z \mid W_Y}}{\sqrt{(1-R^2_{W_Z \sim Z \mid W_Y})(1-R^2_{U \sim Z \mid W_Y})}} = - \frac{R_{W_Z \sim Z \mid W_Y} R_{U \sim Z \mid W_Y}}{\sqrt{(1-R^2_{W_Z \sim Z \mid W_Y})(1-R^2_{U \sim Z \mid W_Y})}}.
\end{align*} 
This allows us to re-express $1- R^2_{W_Z \sim U \mid W_Y, Z}$ as
\begin{align} 
1- R_{W_Z \sim U \mid W_Y, Z}^2 &=1 - \frac{R_{W_Z \sim Z \mid W_Y}^2 R_{U \sim Z \mid W_Y}^2}{(1-R^2_{W_Z \sim Z \mid W_Y})(1-R^2_{U \sim Z \mid W_Y})} \nonumber \\
&= \frac{(1-R^2_{W_Z \sim Z \mid W_Y})(1-R^2_{U \sim Z \mid W_Y}) - R_{W_Z \sim Z \mid W_Y}^2 R_{U \sim Z \mid W_Y}^2}{(1-R^2_{W_Z \sim Z \mid W_Y})(1-R^2_{U \sim Z \mid W_Y})}  \nonumber \\
&= \frac{1 - R^2_{U \sim Z \mid W_Y} - R^2_{W_Z \sim Z \mid W_Y}}{(1-R^2_{W_Z \sim Z \mid W_Y})(1-R^2_{U \sim Z \mid W_Y})}   \nonumber \\
&= \frac{1 - R^2_{Z \sim U \mid W_Y} - R^2_{W_Z \sim Z \mid W_Y}}{(1-R^2_{W_Z \sim Z \mid W_Y})(1-R^2_{Z \sim U \mid W_Y})}  \label{eqn:r2wzu}
\end{align} 

We will now derive two expressions for $R^2_{Z \sim W_Z, U \mid W_Y}$ that can be used to simplify the expression. First, note
\begin{align*} 
\var(Z^{\bot U, W_Z, W_Y}) &= \var(Z^{\bot W_Y}) \frac{\var(Z^{\bot W_Z, W_Y})}{\var(Z^{\bot W_Y})} \cdot \frac{\var(Z^{\bot W_Y, W_Z, U})}{\var(Z^{\bot W_Z, W_Y})}\\
&= \var(Z^{\bot W_Y}) \cdot (1-R^2_{Z \sim W_Z \mid W_Y}) \cdot (1-R^2_{Z \sim U \mid W_Z, W_Y})
\end{align*} 
As a result, $R^2_{Z \sim W_Z, U \mid W_Y}$ can be written as: 
\begin{align} 
R^2_{Z \sim W_Z, U \mid W_Y} &= 1- \frac{\var(Z^{\bot U, W_Z, W_Y})}{\var(Z^{\bot W_Y})} \nonumber \\
&= 1 - (1-R^2_{Z \sim W_Z \mid W_Y}) \cdot (1-R^2_{Z \sim U \mid W_Z, W_Y}) \nonumber \\
&= 1 - (1 - R^2_{Z \sim U \mid W_Z, W_Y} - R^2_{Z \sim W_Z \mid W_Y} + R^2_{Z \sim W_Z \mid W_Y} \cdot R^2_{Z \sim U \mid W_Z, W_Y}) \nonumber \\
&= R^2_{Z \sim U \mid W_Z, W_Y} + R^2_{Z \sim W_Z \mid W_Y} - R^2_{Z \sim W_Z \mid W_Y} \cdot R^2_{Z \sim U \mid W_Z, W_Y} \nonumber \\
&= R^2_{Z \sim W_Z \mid W_Y} + R^2_{Z \sim U \mid W_Z, W_Y} \left( 1- R^2_{Z \sim W_Z \mid W_Y}\right)
\label{eqn:R2_z_wz_u_1}
\end{align} 

We can equivalently write $R^2_{Z \sim W_Z, U \mid W_Y}$ as a function of the linear combination of $W_Z^{\bot W_Y}$ and $U^{\bot W_Y}$: 
\begin{align*} 
R^2_{Z \sim W_Z, U \mid W_Y} &= \frac{\var(\hat c_{w_z} W_Z^{\bot W_Y} + \hat c_{u} U^{\bot W_Y})}{\var(Z^{\bot W_Y})},
\end{align*} 
where $\hat c_{w_z}$ and $\hat c_{u}$ represent the coefficients from the population-level regression, regressing $Z$ with $W_Z$ and $U$, after partialing out dependency with $W_Y$. Under the assumption of exogeneity, this expression can be simplified. In particular, applying FWL, $\hat c_{w_z}$ and $\hat c_u$ are equivalent to the coefficients obtained from running a regression of $Z^{\bot W_Y}$ with $W_Z^{\bot W_Y}$ and a regression of $Z^{\bot W_Y}$ with $U^{\bot W_Y}$, separately. As a result, 
\begin{align*} 
\var&(\hat c_{w_z} W_Z^{\bot W_Y} + \hat c_{u} U^{\bot W_Y}) \\
&= \hat c_{w_z}^2 \var(W_Z^{\bot W_Y}) + \hat c_{u}^2 \var(U^{\bot W_Y}) + 2\hat c_{w_z} \hat c_{u}  \underbrace{\cov(W_Z^{\bot W_Y}, U^{\bot W_Y})}_{=0 \text{ under exogeneity}}\\
&= \left( \frac{\cov(W_Z^{\bot U, W_Y}, Z^{\bot U, W_Y})}{\var(W_Z^{\bot U, W_Y})} \right)^2 \var(W_Z^{\bot W_Y})  + \left( \frac{\cov(U^{\bot W_Z, W_Y}, Z^{\bot W_Z, W_Y})}{\var(U^{\bot W_Z, W_Y})}\right)  \var(U^{\bot W_Y}) \\
&= \left( \frac{\cov(W_Z^{\bot W_Y}, Z^{\bot W_Y})}{\var(W_Z^{\bot W_Y})} \right)^2 \var(W_Z^{\bot W_Y})  + \left( \frac{\cov(U^{\bot W_Y}, Z^{\bot W_Y})}{\var(U^{\bot W_Y})}\right)  \var(U^{\bot W_Y})\\
&= \cor(W_Z^{\bot W_Y}, Z^{\bot W_Y})^2 \var(Z^{\bot W_Y}) + \cor(U^{\bot W_Y}, Z^{\bot W_Y})^2 \var(Z^{\bot W_Y})
\end{align*} 
As a result, we can express $R^2_{Z \sim W_Z, U \mid W_Y}$ as
\begin{equation} 
R^2_{Z \sim W_Z, U \mid W_Y} = R^2_{Z \sim W_Z \mid W_Y} + R^2_{Z \sim U \mid W_Y}
\label{eqn:R2_z_wz_u_2}
\end{equation} 
Setting Eq. \ref{eqn:R2_z_wz_u_1} equal to Eq. \ref{eqn:R2_z_wz_u_2}, 
\begin{align} 
R^2_{Z \sim W_Z \mid W_Y} + R^2_{Z \sim U \mid W_Z, W_Y} \left( 1- R^2_{Z \sim W_Z \mid W_Y}\right)  &= R^2_{Z \sim W_Z \mid W_Y} + R^2_{Z \sim U \mid W_Y} \nonumber \\
\implies R^2_{Z\sim U \mid W_Y} &= R^2_{Z \sim U \mid W_Z, W_Y} (1-R^2_{Z \sim W_Z \mid W_Y})
\label{eqn:r2_z_u}
\end{align}
Finally, substituting Eq. \ref{eqn:r2_z_u} into Eq. \ref{eqn:r2wzu}, 
\begin{align} 
1-R^2_{W_Z \sim U \mid W_Y, Z}
 &= \frac{1 - R^2_{Z \sim U \mid W_Y} - R^2_{W_Z \sim Z \mid W_Y}}{(1-R^2_{W_Z \sim Z \mid W_Y})(1-R^2_{Z \sim U \mid W_Y})} \nonumber \\
 &= \frac{1 - R^2_{W_Z \sim Z \mid W_Y} \cdot \frac{1}{1-R^2_{Z \sim U \mid W_Y}}}{1-R^2_{W_Z \sim Z \mid W_Y}} \nonumber \\
 &= \frac{1 - R^2_{W_Z \sim Z \mid W_Y} \cdot \frac{1}{1-R^2_{Z \sim U \mid W_Z, W_Y} (1-R^2_{Z \sim W_Z \mid W_Y})}}{1-R^2_{W_Z \sim Z \mid W_Y}} \nonumber \\
 &= \frac{1-R^2_{Z \sim U \mid W_Z, W_Y} (1-R^2_{Z \sim W_Z \mid W_Y}) - R^2_{W_Z \sim Z \mid W_Y} \cdot }{(1-R^2_{W_Z \sim Z \mid W_Y})(1-R^2_{Z \sim U \mid W_Z, W_Y} (1-R^2_{Z \sim W_Z \mid W_Y}))} \nonumber \\
 &= \frac{(1-R^2_{W_Z \sim Z \mid W_Y}) (1- R^2_{Z \sim U \mid W_Z, W_Y})}{(1-R^2_{W_Z \sim Z \mid W_Y})(1-R^2_{Z \sim U \mid W_Z, W_Y} (1-R^2_{Z \sim W_Z \mid W_Y}))} \nonumber \\
 &= \frac{1- R^2_{Z \sim U \mid W_Z, W_Y}}{1-R^2_{Z \sim U \mid W_Z, W_Y} (1-R^2_{Z \sim W_Z \mid W_Y})} 
 \label{eqn:r2wzu2}
\end{align} 

Replacing in Eq. \ref{eqn:iv-exog1} with Eq. \ref{eqn:r2wzu2} provides us with the bias expression: 
\begin{align*}
&\frac{1}{R_{Z \sim W_Z \mid W_Y}} \cdot \frac{\sd(Y^{\bot W_Y, Z})}{\sd(Z^{\bot W_Y})} \times R_{Y \sim W_Z \mid W_Y, Z, U} \cdot \sqrt{\left( 1-R^2_{Y \sim U \mid W_Y, Z}\right) \cdot \frac{1}{(1-R^2_{Z \sim W_Z \mid W_Y}) \cdot (1-R^2_{W_Z \sim U \mid W_Y, Z})}}\\
&= \frac{1}{R_{Z \sim W_Z \mid W_Y}} \cdot \frac{\sd(Y^{\bot W_Y, Z})}{\sd(Z^{\bot W_Y})} \times R_{Y \sim W_Z \mid W_Y, Z, U} \cdot \sqrt{\left( 1-R^2_{Y \sim U \mid W_Y, Z}\right) \cdot \frac{1-R^2_{Z \sim U \mid W_Z, W_Y} (1-R^2_{Z \sim W_Z \mid W_Y})}{(1-R^2_{Z \sim W_Z \mid W_Y}) \cdot (1- R^2_{Z \sim U \mid W_Z, W_Y})}}
\end{align*} 

To derive the relative bias condition, we solve for the $R^2_{Y \sim W_Z \mid W_Y, Z, U}$ where $\bias(\tau_\iv) \leq \bias(\tau_\soo)$. 
\begin{align}
R^2_{Y \sim W_Z \mid W_Y, Z, U} \leq& \frac{R^2_{Y \sim U \mid W_Y, W_Z, Z}}{1- R^2_{Y \sim U \mid W_Y, Z}} \cdot \frac{R^2_{Z \sim U \mid W_Y, W_Z}}{1-R^2_{Z \sim U \mid W_Y, W_Z}} \cdot \left( 1- R^2_{Z \sim W_Z \mid W_Y} \right) \cdot \left( 1- R^2_{W_Z \sim U \mid W_Y, Z} \right) \cdot \frac{\var(Z^{\bot W_Y})}{\var(Y^{\bot W_Y, Z})} \nonumber \\
&\times R^2_{Z \sim W_Z \mid W_Y} \cdot \frac{\var(Y^{\bot W_Y, W_Z, Z})}{\var(Z^{\bot W_Y, W_Z})}  \nonumber \\
=& \frac{R^2_{Y \sim U \mid W_Y, W_Z, Z}}{1- R^2_{Y \sim U \mid W_Y, Z}} \cdot \frac{R^2_{Z \sim U \mid W_Y, W_Z}}{1-R^2_{Z \sim U \mid W_Y, W_Z}} \cdot \left( 1- R^2_{Z \sim W_Z \mid W_Y} \right) \cdot \left( 1- R^2_{W_Z \sim U \mid W_Y, Z} \right)  \nonumber \\
&\times \frac{ \var(Y^{\bot W_Y, W_Z, Z})}{\var(Y^{\bot W_Y, Z})} \cdot R^2_{Z \sim W_Z \mid W_Y} \cdot\frac{\var(Z^{\bot W_Y})}{\var(Z^{\bot W_Y, W_Z})}  \nonumber \\
=& \frac{R^2_{Y \sim U \mid W_Y, W_Z, Z}}{1- R^2_{Y \sim U \mid W_Y, Z}} \cdot \frac{R^2_{Z \sim U \mid W_Y, W_Z}}{1-R^2_{Z \sim U \mid W_Y, W_Z}} \cdot \left( 1- R^2_{Z \sim W_Z \mid W_Y} \right) \cdot  \left( 1- R^2_{W_Z \sim U \mid W_Y, Z} \right) \nonumber \\
&\times (1-R^2_{Y \sim W_Z \mid W_Y, Z}) \cdot  \frac{R^2_{Z \sim W_Z \mid W_Y}}{1-R^2_{Z \sim W_Z \mid W_Y}}   \nonumber \\
=& \frac{R^2_{Y \sim U \mid W_Y, W_Z, Z}}{1- R^2_{Y \sim U \mid W_Y, Z}} \cdot \frac{R^2_{Z \sim U \mid W_Y, W_Z}}{1-R^2_{Z \sim U \mid W_Y, W_Z}} \cdot (1-R^2_{Y \sim W_Z \mid W_Y, Z}) \cdot  \left( 1- R^2_{W_Z \sim U \mid W_Y, Z} \right) \cdot R^2_{Z \sim W_Z \mid W_Y}  \nonumber \\
=& \frac{R^2_{Y \sim U \mid W_Y, W_Z, Z}}{1- \gamma \cdot R^2_{Y \sim U \mid W_Y, W_Z, Z}} \cdot \frac{R^2_{Z \sim U \mid W_Y, W_Z}}{1-R^2_{Z \sim U \mid W_Y, W_Z}} \cdot (1-R^2_{Y \sim W_Z \mid W_Y, Z}) \cdot  \left( 1- R^2_{W_Z \sim U \mid W_Y, Z} \right)\cdot  R^2_{Z \sim W_Z \mid W_Y},
\label{eqn:rel-bias-cond1}
\end{align}
where $\gamma =  R^2_{Y \sim U \mid W_Y, Z}/R^2_{Y \sim U \mid W_Y, W_Z, Z}$. 

Replacing $1-R^2_{W_Z \sim U \mid W_Y, Z}$ in Eq. \ref{eqn:rel-bias-cond1} with Eq. \ref{eqn:r2wzu2} completes the proof. 
\end{proof}

\subsection{\texorpdfstring{Theorem~\ref{thm-prox-bias}}{Theorem~}}\label{theorem-thm-prox-bias}

\begin{proof}
As in the main text, use Frisch-Waugh-Lovell to write
\(\tau_{\prox}-\tau\) as \begin{align*}
\tau_{\prox} -\tau &=
    \underbrace{\beta_u \frac{\cov(U^{\bot \hat W_Y}, Z^{\bot \hat W_Y})}{\var(Z^{\bot \hat W_Y})}}_{(1)} +
    \underbrace{\beta_{w_y} \frac{\cov(W_Y^{\bot \hat W_Y}, Z^{\bot \hat W_Y})}{\var(Z^{\bot \hat W_Y})}}_{(2)} +
    \underbrace{\beta_{w_z} \frac{\cov(W_Z^{\bot \hat W_Y}, Z^{\bot \hat W_Y})}{\var(Z^{\bot \hat W_Y})}}_{(3)}.
\end{align*}

\noindent To start, we derive an expression for \(Z^{\bot \hat W_Y}\).
By construction, because \(\hat W_Y\) is a linear combination of \(Z\)
and \(W_Z\), there exists some \(\hat b_z\) and \(\hat b_{w_z}\) such
that: \begin{equation}\protect\phantomsection\label{eq-z-bot}{
Z^{\bot \hat W_Y} = Z - \hat b_z Z - \hat b_{w_z} W_Z = (1-\hat b_z) Z - \hat b_{w_z} W_Z.
}\end{equation} From Lemma~\ref{lem-wy-y},
\(\cov(W_Y^{\bot \hat W_Y}, Z^{\bot \hat W_Y}) = 0\). As a result,
re-writing \(\cov(W_Y^{\bot \hat W_Y}, Z^{\bot \hat W_Y})\) with
Eq.~\ref{eq-z-bot} allows us to write \(\hat b_{w_z}\) as a function of
covariances. More specifically, \begin{align*} 
\cov(W_Y^{\bot \hat W_Y}, Z^{\bot \hat W_Y})  &= \cov(W_Y, Z^{\bot \hat W_Y}) - \hat b_{w_y} \underbrace{\cov(\hat W_Y, Z^{\bot \hat W_Y})}_{=0} \\
&=  \cov(W_Y, Z - \hat b_z Z - \hat b_{w_z} W_Z) \\
&= (1- \hat b_z) \cov(W_Y, Z)- \hat b_{w_z} \cov(W_Y, W_Z).
\end{align*} As a result, \[
1- \hat b_z = \frac{\hat b_{w_z}}{\cov(W_Y, Z)} \cov(W_Y, W_Z), \quad\text{ and } \quad \hat b_{w_z} = \frac{1- \hat b_z}{\cov(W_Y, W_Z)} \cov(W_Y, Z).
\]

Then, \begin{align*} 
Z^{\bot \hat W_Y} &= (1-\hat b_z) Z - \hat b_{w_z} W_Z \\
&= \frac{\hat b_{w_z}}{\cov(W_Y, Z)} \cov(W_Y, W_Z) Z- \hat b_{w_z} W_Z \\
&= \frac{\hat b_{w_z}}{\cov(W_Y, Z)} \left( \cov(W_Y, W_Z) \cdot Z - \cov(W_Y, Z) \cdot W_Z\right) 
\end{align*} For ease of notation, we write
\(k := \hat b_{w_z}/\cov(W_Y, Z)\).

\begin{align*}
\frac{1}{k}\cov(U^{\bot \hat W_Y}, Z^{\bot \hat W_Y})
=& \cov(W_Y, W_Z)\cov(U, Z) - \cov(W_Y, Z)\cov(U, W_Z)\\
=& \cov(W_Y, W_Z)\left(\gamma_u \var(U) + \gamma_{w_y}\cov(U, W_Y) + \gamma_{w_z}\cov(U, W_Z)\right) \\
&- \left(\gamma_u \cov(U, W_Y) + \gamma_{w_y}\var(W_Y) + \gamma_{w_z}\cov(W_Y, W_Z)\right)\cov(U, W_Z) \\
=& \cov(W_Y, W_Z)\left(\gamma_u \var(U) + \gamma_{w_y}\cov(U, W_Y)\right)
 - \left(\gamma_u \cov(U, W_Y) + \gamma_{w_y}\var(W_Y)\right)\cov(U, W_Z),
\end{align*} where the \(\gamma_{w_z}\) terms cancel. Now substituting
in \(\cov(W_Y, W_Z) = \varphi_u\cov(U, W_Y) + \varphi_{w_y}\var(W_Y)\)
and \(\cov(U, W_Z) = \varphi_u\var(U) + \varphi_{w_y}\cov(U, W_Y)\):
\begin{align*}
&\left(\varphi_u\cov(U, W_Y) + \varphi_{w_y}\var(W_Y)\right)\left(\gamma_u \var(U) + \gamma_{w_y}\cov(U, W_Y)\right) \\
&- \left(\gamma_u \cov(U, W_Y) + \gamma_{w_y}\var(W_Y)\right)\left(\varphi_u\var(U) + \varphi_{w_y}\cov(U, W_Y)\right) \\
=& \varphi_u\gamma_u\cov(U,W_Y)\var(U) + \varphi_u\gamma_{w_y}\cov(U,W_Y)^2
 + \varphi_{w_y}\gamma_u\var(W_Y)\var(U) + \varphi_{w_y}\gamma_{w_y}\var(W_Y)\cov(U,W_Y) \\
&- \gamma_u\varphi_u\cov(U,W_Y)\var(U) - \gamma_u\varphi_{w_y}\cov(U,W_Y)^2
 - \gamma_{w_y}\varphi_u\var(W_Y)\var(U) - \gamma_{w_y}\varphi_{w_y}\var(W_Y)\cov(U,W_Y) \\
=& (\varphi_{w_y}\gamma_u - \gamma_{w_y}\varphi_u)\var(W_Y)\var(U)
 + (\varphi_u\gamma_{w_y} - \gamma_u\varphi_{w_y})\cov(U, W_Y)^2 \\
=& (\varphi_{w_y}\gamma_u - \gamma_{w_y}\varphi_u)\left(\var(W_Y)\var(U) - \cov(U, W_Y)^2\right)
\end{align*} Finally, since \(W_Y = \alpha_u U + \varepsilon_{w_y}\),
\begin{align*} 
\var(W_Y)\var(U) - \cov(U, W_Y)^2 &= \var(W_Y) \var(U) - \alpha_u^2 \var(U)^2 = \var(U) \left( \var(W_Y) - \alpha_u^2 \var(U) \right).
\end{align*} As a result,
\begin{equation}\protect\phantomsection\label{eq-cov-u}{
\cov(U^{\bot \hat W_Y}, Z^{\bot \hat W_Y}) = k(\varphi_{w_y}\gamma_u - \gamma_{w_y}\varphi_u) \left( \var(W_Y) - \alpha_u^2 \var(U) \right) \var(U)
}\end{equation}

We can show that \(\var(Z^{\bot \hat W_Y})\) can similarly be expanded:
\begin{align} 
\var(Z^{\bot \hat W_Y}) &= \cov(Z^{\bot \hat W_Y}, Z^{\bot \hat W_Y}) \nonumber \\
&= \cov(Z, Z^{\bot \hat W_Y}) \nonumber\\
&= \cov(Z, k (\cov(W_Y, W_Z) Z - \cov(W_Y, Z) W_Z)) \nonumber \\
&= k \left( \cov(W_Y, W_Z) \var(Z) - \cov(W_Y, Z) \cov(W_Z, Z)\right)  \label{eq:var-z}
\end{align}

Then, plugging in Eq.~\ref{eq-cov-u} and Eq. \ref{eq:var-z}:
\begin{align*}
\frac{\cov(U^{\bot \hat W_Y}, Z^{\bot \hat W_Y})}{\var(Z^{\bot \hat W_Y})} &= (\varphi_{w_y}\gamma_u - \gamma_{w_y}\varphi_u)\frac{ \left( \var(W_Y) - \alpha_u^2 \var(U) \right) \var(U)}{\cov(W_Y, W_Z) \var(Z) - \cov(W_Y, Z) \cov(W_Z, Z)}
\end{align*} Substituting into term (1), and defining
\(S_{W_Z,Z}\coloneq \cor(W_Z^{\bot \hat W_Y}, Z^{\bot \hat W_Y})\),
\begin{align*}
\tau_{\prox} - \tau &= \beta_u \frac{(\varphi_{w_y}\gamma_u - \gamma_{w_y}\varphi_u) \left( \var(W_Y) - \alpha_u^2 \var(U) \right) \var(U)}{\cov(W_Y, W_Z) \var(Z) - \cov(W_Y, Z) \cov(W_Z, Z)}
 + \beta_{w_z} \frac{\cov(W_Z^{\bot \hat W_Y}, Z^{\bot \hat W_Y})}{\var(Z^{\bot \hat W_Y})}\\
&=  S_{W_Z, Z}\cdot \beta_{w_z} \frac{\sd(W_Z^{\bot \hat W_Y})}{\sd(Z^{\bot \hat W_Y})} +  \beta_u (\varphi_{w_y}\gamma_u - \gamma_{w_y}\varphi_u)\frac{ \left( \var(W_Y) - \alpha_u^2 \var(U) \right) \var(U)}{\cov(W_Y, W_Z) \var(Z) - \cov(W_Y, Z) \cov(W_Z, Z)} \\
&=  S_{W_Z, Z}\cdot \beta_{w_z} \frac{\sd(W_Z^{\bot \hat W_Y})}{\sd(Z^{\bot \hat W_Y})} +  \beta_u (\varphi_{w_y}\gamma_u - \gamma_{w_y}\varphi_u)\frac{ \left( 1- R^2_{W_Y \sim U} \right) \var(W_Y) \var(U)}{\cov(W_Y, W_Z) \var(Z) - \cov(W_Y, Z) \cov(W_Z, Z)} 
\end{align*}

Notice that because \(\hat W_Y\) is a linear combination of \(W_Z\) and
\(Z\), we must have
\(|\cor(W_Z^{\bot \hat W_Y}, Z^{\bot \hat W_Y})|=1\), so
\(S_{W_Z,Z}\in\{-1, 1\}\).
\end{proof}

\subsection{\texorpdfstring{Corollary~\ref{cor-prox-r2}}{Corollary~}}\label{corollary-cor-prox-r2}

\begin{proof}
We will derive \(R^2\) expressions for both terms. \[
\underbrace{S_{W_Z, Z}\cdot \beta_{w_z} \frac{\sd(W_Z^{\bot \hat W_Y})}{\sd(Z^{\bot \hat W_Y})}}_{\text{Term 1}} +  \underbrace{\beta_u (\varphi_{w_y}\gamma_u - \gamma_{w_y}\varphi_u)\frac{ \left( 1- R^2_{W_Y \sim U} \right) \var(W_Y) \var(U)}{\cov(W_Y, W_Z) \var(Z) - \cov(W_Y, Z) \cov(W_Z, Z)}}_{\text{Term 2}}
\] For Term 1, we leverage the expression from \eqref{eqn:beta_wz_r2_2}
to re-write \(\beta_{w_z}\): \begin{align*}
S_{W_Z, Z} \cdot \beta_{w_z} \cdot \frac{\sd(W_Z^{\bot \hat W_Y})}{\sd(Z^{\bot \hat W_Y})} &= S_{W_Z, Z}  \cdot R_{Y \sim W_Z \mid W_Y, Z, U} \sqrt{\frac{\var(Y^{\bot W_Y, Z})}{\var(W_Z^{\bot W_Y})} \cdot \frac{(1-R^2_{Y \sim U \mid W_Y, Z})}{(1-R^2_{W_Z \sim U \mid W_Y, Z})(1-R^2_{Z \sim W_Z \mid W_Y})}} \cdot \frac{\sd(W_Z^{\bot \hat W_Y})}{\sd(Z^{\bot \hat W_Y})} \\
&=  S_{W_Z, Z}  \cdot R_{Y \sim W_Z \mid W_Y, Z, U} \sqrt{\frac{\var(Y^{\bot W_Y, Z})}{\var(Z^{\bot \hat W_Y})} \cdot \frac{\var(W_Z^{\bot \hat W_Y})}{\var(W_Z^{\bot W_Y})} \cdot \frac{(1-R^2_{Y \sim U \mid W_Y, Z})}{(1-R^2_{W_Z \sim U \mid W_Y, Z})(1-R^2_{Z \sim W_Z \mid W_Y})}}
\end{align*}

For term 2, we will derive \(R^2\) expressions for each coefficient.
First, we re-write \(\beta_u\) as: \begin{align}
\beta_u &= \frac{\cov(U^{\bot W_Z, W_Y, Z}, Y^{\bot W_Z, W_Y, Z})}{\var(U^{\bot W_Z, W_Y, Z})} \nonumber \\
&= \cor(U^{\bot W_Z, W_Y, Z}, Y^{\bot W_Z, W_Y, Z}) \cdot \frac{\sd(Y^{\bot W_Z, W_Y, Z})}{\sd(U^{\bot W_Z, W_Y, Z})} \nonumber \\
&= R_{Y \sim U \mid W_Z, W_Y, Z} \sqrt{\frac{\var(Y)}{\var(U)} \cdot \frac{1-R^2_{Y \sim W_Y, W_Z, Z}}{(1-R^2_{U \sim W_Z \mid W_Y, Z})(1-R^2_{U \sim W_Y, Z})}} \nonumber \\
&= R_{Y \sim U \mid W_Z, W_Y, Z} \sqrt{\frac{\var(Y)}{\var(U)} \cdot \frac{1-R^2_{Y \sim W_Y, W_Z, Z}}{(1-R^2_{W_Z \sim U \mid W_Y, Z})(1-R^2_{U \sim W_Y, Z})}}
\label{eqn:beta_u_r2}
\end{align}

Then, \begin{align*} 
\varphi_{w_y}\gamma_u - \gamma_{w_y}\varphi_u =& \frac{\cov(W_Y^{\bot U}, W_Z^{\bot U})}{\var(W_Y^{\bot U})} \cdot \frac{\cov(U^{\bot W_Y, W_Z}, Z^{\bot W_Y, W_Z})}{\var(U^{\bot W_Y, W_Z})} - \frac{\cov(W_Y^{\bot U, W_Z}, Z^{\bot U, W_Z})}{\var(W_Y^{\bot U, W_Z})} \frac{\cov(U^{\bot W_Y}, W_Z^{\bot W_Y})}{\var(U^{\bot W_Y})} \\
=& \cor(W_Y^{\bot U}, W_Z^{\bot U}) \frac{\sd(W_Z^{\bot U})}{\sd(W_Y^{\bot U})} \cdot \cor(U^{\bot W_Y, W_Z}, Z^{\bot W_Y, W_Z}) \frac{\sd(Z^{\bot W_Y, W_Z})}{\sd(U^{\bot W_Y, W_Z})} \\
&- \cor(W_Y^{\bot U, W_Z}, Z^{\bot U, W_Z}) \frac{\sd(Z^{\bot U, W_Z})}{\sd(W_Y^{\bot U, W_Z})} \cor(U^{\bot W_Y}, W_Z^{\bot W_Y}) \frac{\sd(W_Z^{\bot W_Y})}{\sd(U^{\bot W_Y})} \\
=& R_{W_Y \sim W_Z \mid U} \cdot R_{Z \sim U \mid W_Y, W_Z} \frac{\sd(W_Z^{\bot U})}{\sd(W_Y^{\bot U})} \cdot  \frac{\sd(Z^{\bot W_Y, W_Z})}{\sd(U^{\bot W_Y, W_Z})} - R_{W_Y \sim Z \mid W_Z, U}  R_{W_Z \sim U \mid W_Y} \frac{\sd(Z^{\bot U, W_Z})}{\sd(W_Y^{\bot U, W_Z})} \frac{\sd(W_Z^{\bot W_Y})}{\sd(U^{\bot W_Y})} 
\end{align*}

Combining everything together:

\begin{align*} 
\bias(\tau_\prox) =&  S_{W_Z, Z}  \cdot R_{Y \sim W_Z \mid W_Y, Z, U} \sqrt{\frac{\var(Y^{\bot W_Y, Z})}{\var(Z^{\bot \hat W_Y})} \cdot \frac{\var(W_Z^{\bot \hat W_Y})}{\var(W_Z^{\bot W_Y})} \cdot \frac{(1-R^2_{Y \sim U \mid W_Y, Z})}{(1-R^2_{W_Z \sim U \mid W_Y, Z})(1-R^2_{Z \sim W_Z \mid W_Y})}} \\
&+ R_{Y \sim U \mid W_Z, W_Y, Z} \sqrt{\frac{\var(Y)}{\var(U)} \cdot \frac{1-R^2_{Y \sim W_Y, W_Z, Z}}{(1-R^2_{W_Z \sim U \mid W_Y, Z})(1-R^2_{U \sim W_Y, Z})}} \\
& \times \left( R_{W_Y \sim W_Z \mid U} \cdot R_{Z \sim U \mid W_Y, W_Z} \frac{\sd(W_Z^{\bot U})}{\sd(W_Y^{\bot U})} \cdot  \frac{\sd(Z^{\bot W_Y, W_Z})}{\sd(U^{\bot W_Y, W_Z})} - R_{W_Y \sim Z \mid W_Z, U}  R_{W_Z \sim U \mid W_Y} \frac{\sd(Z^{\bot U, W_Z})}{\sd(W_Y^{\bot U, W_Z})} \frac{\sd(W_Z^{\bot W_Y})}{\sd(U^{\bot W_Y})} \right) \\
& \times \frac{ \left( 1- R^2_{W_Y \sim U} \right) \var(W_Y) \var(U)}{\cov(W_Y, W_Z) \var(Z) - \cov(W_Y, Z) \cov(W_Z, Z)}
\end{align*}
\end{proof}

\end{document}